\documentclass[12pt]{article}
\usepackage[dvips]{graphicx}
\usepackage{amsmath, amssymb}
\usepackage{amsthm}

\newtheorem{proposition}{Proposition}[section]

\newcommand{\pweight}{{Q}} %
\newcommand{\relw}{\nu}    %
\newcommand{\rnp}{\rho}    %
\newcommand{\prc}{\theta}

\newcommand{\cP}{{\cal P}}
\newcommand{\cK}{{\cal K}}

\newlength{\IndentI}
\newlength{\IndentII}
\newlength{\IndentIII}
\newlength{\WidthI}
\newlength{\WidthII}
\newlength{\WidthIII}
\title{New procedures for testing whether stock price processes are martingales}

\author{
  Kei Takeuchi
\footnote{
  Emeritus, Graduate School of Economics, 
  University of Tokyo}
\\
Akimichi Takemura
\footnote{
Graduate School of Information Science and Technology,
University of Tokyo, \ 
7-3-1 Hongo, Bunkyo-ku, Tokyo 113-8656, JAPAN,  \ 
takemura@stat.t.u-tokyo.ac.jp}
\\
Masayuki Kumon\footnote{masayuki\_kumon@smile.odn.ne.jp}
}

\date{July 2009}

\begin{document}
\maketitle

\begin{abstract}
  We propose procedures for testing whether stock price processes 
  are martingales   based on limit order type betting strategies.  We
  first show that the null hypothesis of martingale property of a
  stock price process can be tested based on the capital process of
  a betting strategy.  In particular with high frequency Markov type strategies we
  find that martingale null hypotheses are rejected for many stock 
  price processes.
\end{abstract}


\section{Introduction}
\label{sec:intro}

The efficient market hypothesis (EMH), that no one with finite capital can
consistently outperform the market,
is the fundamental assumption in the theory
of financial engineering.  In mathematical finance the efficient market
hypothesis is formulated as the martingale property of 
price processes  of tradable assets such as stocks.  

Often the martingale assumption is replaced by a more convenient
assumption of ``random walk''.  Although exact formulation of random
walk depends on literature (e.g.\ \cite[Chapter 2]{campbell-lo-mackinlay},
\cite{beechey-et-al}), the usual assumption is that 
the price process observed at equi-spaced time points
has independent increments with mean zero, 
after adjustment of the systematic trend.
However there are many
empirical studies showing that stock price processes are not random walks
(e.g.\ \cite{campbell-lo-mackinlay}, \cite{lo-mackinlay}).

Note that the class of martingales is
larger than the class of random walks with zero expected increment (\cite{leroy-1989}).  
This implies that rejecting the hypothesis of random walk does not necessarily mean
rejecting the hypothesis of martingale.  Therefore it is desirable to
{\em directly} test the martingale assumption of stock price processes without assuming 
a random walk.

We propose to test the martingale hypothesis of price processes based on
our previous works 
on ``limit order'' type betting strategies 
(\cite{takeuchi-kumon-takemura-bernoulli,
takeuchi-kumon-takemura-markov})
in the framework of game-theoretic probability by Shafer and Vovk (\cite{shafer/vovk:2001}) and an adaptation of the result by Dubins
and Schwarz \cite{dubins-schwarz} to positive  measure-theoretic 
martingales.
As discussed in Section
\ref{subsec:nonnegative-martingale-and-test},
betting strategies in game-theoretic probability naturally
yield sequential testing procedures for the measure-theoretic  martingale
hypothesis.  Moreover our strategies in 
\cite{takeuchi-kumon-takemura-bernoulli,
takeuchi-kumon-takemura-markov} are of very simple form
and provide convenient testing procedures.  

Our testing procedures depend on the direction of a price process
(``ups'' and ``downs'') at times, when
the process hits fixed horizontal grids of prices.   We call these time points
{\em hitting times}. Thus our procedures are
very much different from procedures based on increments of price processes
observed at equi-spaced time points.

One advantage of our testing procedure (compared to equi-spaced
procedures) is that we do not have to worry about specifications
of distributions of the increments, such as the the heaviness of the
tail of the distribution of the increments, because 
in our procedures the amount of the increments are fixed by 
the given grids of prices.  Our procedure depends only on 
the ups or the downs of the price process between hitting times.
This is contrasted with problems of model specifications for testing EMH in 
standard approaches (e.g.\ \cite{timmermann-granger}).

Our approach is similar to
the approach of testing EMH based
on algorithmic complexity of the ups and downs of price processes
in \cite{shmilovici-etal} and \cite{giglio-etal}.  In fact, it is well 
known that betting 
strategies and compression algorithms of binary strings are essentially equivalent
\cite[Chapter 6]{cover-thomas:2006}.
However the approaches of \cite{shmilovici-etal} and \cite{giglio-etal} are based
on price movements observed at equi-spaced time points and therefore
they are different from ours.

The organization of this paper is as follows.
In Section \ref{sec:preliminary} we summarize results on betting
strategies in the framework of game-theoretic probability of Shafer and Vovk
(\cite{shafer/vovk:2001}) and
explain that these betting strategies naturally lead to sequential
testing procedures for the null hypothesis of measure-theoretic martingale.
In Section \ref{sec:martingale} we propose our procedure 
for testing
martingale properties of price processes 
based on limit order type betting strategies.
We give numerical results of testing martingale properties of some Japanese
stock price processes  in Section \ref{sec:numerical}, 
where numerical examinations are also added for the effect of transaction costs on capital processes. 
We conclude our paper with some remarks in Section \ref{sec:discussion}.

\section{Preliminary results}
\label{sec:preliminary}
In this section we give an exposition on betting strategies
in game-theoretic probability based on our earlier  results in \cite{ktt:2007b},
\cite{takeuchi-kumon-takemura-bernoulli} and
\cite{takeuchi-kumon-takemura-markov}.  In Section 
\ref{subsec:nonnegative-martingale-and-test}  we also remark the important
fact that these betting strategies can be used for testing the null hypothesis
of measure-theoretic martingale.

\subsection{Prudent strategies in game-theoretic probability and
  a sequential test of measure-theoretic martingale hypothesis}
\label{subsec:nonnegative-martingale-and-test}

For simplicity of exposition we consider the biased-coin tossing game.
This game is a discrete time game 
played by two players called ``Investor'' and ``Market''.
Investor enters the game with the initial capital of $\cK_0=1$.
For each round Investor first decides how much to bet  and then
Market (after seeing the Investor's move) decides the outcome.
In the biased-coin tossing game, the outcome chosen by Market is either 1 (``up'')
or 0 (``down'').  
The formal protocol of the game is written as follows.

\medskip\noindent
\textsc{Biased-Coin Game} \\
\textbf{Protocol:}

\parshape=6
\IndentI   \WidthI
\IndentI   \WidthI
\IndentII  \WidthII
\IndentII  \WidthII
\IndentII  \WidthII
\IndentI   \WidthI
\noindent
${\cal K}_0 =1$ and $0<\rho<1$ are given.\\
FOR  $n=1, 2, \dots$:\\
  Investor announces $\nu_n\in{\mathbb R}$.\\
  Market announces $x_n\in \{0,1\}$.\\
  ${\cal K}_n = {\cal K}_{n-1} + \nu_n (x_n-\rho)$.\\
END FOR

\medskip
We call $\rho$ the risk-neutral probability. %
Investor can choose $\nu_n$ based 
on the past moves $x_1, \dots, x_{n-1}$ of Market. 
Suppose that Investor adopts a strategy $\cP$, which is a function specifying
$\nu_n$ based on $x_1,\dots, x_{n-1}$ with some initial value $\nu_1$.
\[
 \cP: (x_1, \dots, x_{n-1}) \mapsto \nu_n.
\]
Then Investor's capital at the end of round $n$ is written as
\begin{equation}
\label{eq:capital-process}
\cK_n^{\cP}= \cK_0 + \sum_{i=1}^n \cP(x_1, \dots, x_{i-1})(x_i -\rho).
\end{equation}
A betting strategy $\cP$ of Investor is called {\em prudent} if 
Investor is never bankrupt when using $\cP$, i.e.\ 
$\cK_n^{\cP}\ge 0$ for all $n$ and for all $x_1, x_2, \dots \in \{0,1\}$.

In the framework of game-theoretic  probability of Shafer and Vovk (\cite{shafer/vovk:2001}) there is no probabilistic assumption on the behavior of Market.  Therefore
Market in the above biased-coin game may even be adversarial to Investor.
However the usual measure-theoretic assumption on the behavior of Market is
that Market is oblivious to Investor's moves  and chooses $x_n$
independently as $P(x_n=1)=\rho=1-P(x_n=0)$, ignoring Investor's bet $\nu_n$.   
We write this null hypothesis as 
\begin{equation}
\label{eq:null-hypothesis}
H: p_1=\rho, 
\end{equation}
where $p_1=P(x_n=1)$. Note that the mutual independence of $x_n$, $n=1,2,\dots$, is
also implied by $H$ for the biased-coin game, because the outcome $x_n$ is binary.
If $\cP$ is a prudent strategy, then under $H$, 
$\cK^{\cP}_n$ is a usual measure-theoretic non-negative
martingale.  

By the Markov inequality for non-negative measure-theoretic martingales
(cf.\ Chapter II, Section 57 of \cite{rogers-williams-1})
we have the following sequential testing procedure for $H$.

\begin{proposition}
\label{prop:martingale-test}
Let $0 < \alpha < 1$ be given. 
Reject the null hypothesis $H: p_1=\rho$ as soon as
$\cK^{\cP}_n \ge 1/\alpha$, where $\cP$ is a prudent strategy.
This procedure has the significance level $\alpha$.
\end{proposition}

The intuitive interpretation of this proposition is as follows.  $H$ corresponds to 
the efficient market hypothesis (EMH). Investor
can disprove EMH by beating Market, i.e., if he can 
multiply  his capital many times by an appropriate betting strategy.  If Investor
can make his capital 100 times larger than his initial capital $\cK_0$, then
$H$ is rejected at the significance level of 1\%.

We can also use  Ville's inequality (Section 2.5 of \cite{shafer/vovk:2001}, 
page 100 of \cite{ville}), which is now  commonly known as 
Doob's supermartingale inequality ((57.10) of \cite{rogers-williams-1}).  
{}From game-theoretic viewpoint the following procedure
is not very much different from  the procedure in the above proposition,
because it corresponds to stop betting after the hitting
time  $\cK^{\cP}_n \ge 1/\alpha$.
\begin{proposition}
\label{prop:martingale-test-doob}
Let $0 < \alpha < 1$ be and $N>0$ be given. 
Reject the null hypothesis $H: p_1=\rho$ if 
$\max_{0\le n\le N}\cK^{\cP}_n \ge 1/\alpha$, 
where $\cP$ is a prudent strategy.
This procedure has the significance level $\alpha$.
\end{proposition}

We have stated the above propositions for the protocol of
biased-coin game.  However as in \cite{shafer/vovk:2001} we can consider
more complicated protocols, such as the bounded-forecasting game, where
for each round Market chooses $x_n$ in the bounded interval $[0,1]$.
The measure-theoretic interpretation of the null hypothesis $H$ in 
(\ref{eq:null-hypothesis})
is that $x_n-\rho$, $n=1,2,\dots$, are (uniformly bounded) martingale differences.
In this form, the null hypothesis
$H$ does not place any distributional assumptions on $x_n$, except
for the martingale property.  Therefore we are directly testing
the assumption of martingale property. 
As long as $\cP$ is a prudent strategy,
we can test $H$ by Proposition \ref{prop:martingale-test}.

\subsection{Limit order type strategy in continuous time game and embedded coin-tossing game}
\label{subsec:limit-order}

Although the biased-coin game of the previous subsection is very simple, we can
analyze a continuous time game between Investor and Market with the 
biased-coin game of the previous subsection,
by embedding it into continuous time by a limit order type strategy.

Consider a continuous time game between Investor and Market. 
Market chooses a price path $S(t)$, $t\ge 0$, of a financial asset.
We assume that $S(t)$ is continuous and positive.
Investor enters the market at time $t=t_0=0$  
with the initial capital of ${\cal K}(0)=1$ and he will buy or
sell any amount of the asset at discrete time points 
$0=t_0 <  t_1 < t_2 < \cdots$.
Let $M_i \in {\mathbb R}$ denote the amount of the asset he holds
for the time interval $[t_i,t_{i+1})$. 
Then the  capital of Investor ${\cal K}(t)$ at time $t$ is written as
\begin{align}
\label{eq:2-1}
&{\cal K}(0) = 1, \nonumber \\
&{\cal K}(t) = {\cal K}(t_i) + M_i(S(t) - S(t_i))\quad 
\textrm{for}\ \ t_i \le t  < t_{i+1}.
\end{align}
By defining
$\prc_i = M_iS(t_i)/\cK(t_i)$, 
we rewrite (\ref{eq:2-1}) as
\begin{align*}
{\cal K}(t) = {\cal K}(t_i)\left(1 + \prc_i
\frac{S(t) - S(t_i)}{S(t_i)}\right)\quad 
\textrm{for}\ \ t_i \le  t < t_{i+1} .
\end{align*}

In limit order type strategy, 
Investor takes some constant $\delta >0$ and decides the trading times $t_1, t_2, \dots$ as follows. 
After $t_i$ is determined, 
let $t_{i+1}$ be the first time after $t_i$ when either
\begin{align*}
\frac{S(t_{i+1})}{S(t_i)} = 1 + \delta\quad \textrm{or}
\quad = \frac{1}{1 + \delta}
\end{align*}
happens. 
Let $w_i=t_{i+1}-t_i$ denote the {\em waiting times}
between two successive trading times.
In terms of $\log S$, the waiting times $w_i$ are determined by 
\begin{equation}
\label{eq:waiting-time-def}
  \log S(t_{i+1})-\log S(t_i)=\pm \eta,  \quad \eta=\log(1+\delta)
\qquad (\delta=e^\eta-1).
\end{equation}
This process leads to a discrete time coin-tossing game
embedded in the continuous time game in the following manner.  Let
\begin{align}
x_n &= \frac{(1+\delta) S(t_{n+1}) - S(t_n)}
{\delta(2 + \delta) S(t_n)}
= \begin{cases}
 1, & \text{if}\  S(t_{n+1})=S(t_n) (1+\delta), \\
 0, & \text{if}\  S(t_{n+1})=S(t_n)/(1+\delta),
\end{cases} 
\label{eq:def-xn}
\end{align}
and 
\[
\rnp = \frac{1}{2 + \delta}, \quad
\tilde{{\cal K}}_n = {\cal K}(t_{n+1}),\quad 
\nu_n = \frac{\delta(2 + \delta)}
{1 +\delta}\prc_n. 
\]

Now we have the following discrete time coin-tossing game.

\medskip\noindent
\textsc{Embedded Discrete Time Coin-Tossing Game} \\
\textbf{Protocol:}

\parshape=6
\IndentI   \WidthI
\IndentI   \WidthI
\IndentII  \WidthII
\IndentII  \WidthII
\IndentII  \WidthII
\IndentI   \WidthI
\noindent
$\tilde{{\cal K}}_0 :=1$.\\
FOR  $n=1, 2, \dots$:\\
  Investor announces $\nu_n\in{\mathbb R}$.\\
  Market announces $x_n\in \{0, 1\}$.\\
  $\tilde{{\cal K}}_n = \tilde{{\cal K}}_{n-1}(1 + \nu_n (x_n-\rnp))$.\\
END FOR

\medskip
Combining the embedded discrete time coin-tossing game with
Proposition \ref{prop:martingale-test}, we can test whether a
continuous price process chosen by Market is a measure-theoretic
martingale.   
Suppose that the price process $S(t)$ of Market is a positive 
measure-theoretic martingale.  
We write this null hypothesis as
\begin{equation}
\label{eq:continuous-martingale}
\bar H: S(t) \text{\ is a positive martingale.}
\end{equation}
Under $\bar H$, for every $\delta>0$, 
$x_n$'s in the embedded coin-tossing game satisfy
the null hypothesis $H$ in (\ref{eq:null-hypothesis}) with $\rnp=1/(2+\delta)$. 
Therefore any prudent strategy for the embedded coin-tossing game
can be used as a sequential testing procedure of $\bar H$ by Proposition
\ref{prop:martingale-test}.
A particularly useful betting strategy can be given from Bayesian viewpoint as shown 
in the next section.

\subsection{Bayesian betting strategy based on past number of ups and downs}
\label{subsec:bayesian}

A simple betting strategy for the biased-coin game is 
given by a Bayesian consideration (\cite{ktt:2007b}).  
By the embedded coin-tossing game, it can be
applied to the continuous time game.

Let $h_n=n \bar x_n = \sum_{i=1}^n x_i$ denote the number of heads 
and let $t_n=n-h_n$ denote the number 
of tails up to round $n$ in the biased-coin game.
Fix $a>0, b>0$.
We call the following strategy of  Investor a  beta-binomial strategy 
(with the hyperparameters $a, b$):
\begin{equation}
\label{eq:2-5}
\relw_n =\frac{\hat p_n^\pweight-\rnp}{\rnp(1-\rnp)} 
\qquad \text{where}\quad 
\hat p_n^\pweight 
= \frac{a + h_{n-1}}{a + b + n-1}.
\end{equation}
The capital process $\cK_n$ for this strategy 
is explicitly written as 
\begin{equation}
\cK_n  =
\frac{(a)_{h_n} (b)_{t_n}}{(a+b)_n  \rnp^{h_n}
  (1-\rnp)^{t_n}}, 
\label{eq:baysian-capital}
\end{equation}
where
\[
(c)_l=c(c+1)\cdots(c+l-1)= \frac{\Gamma(c+l)}{\Gamma(c)}
\]
for $c>0$  and non-negative integer $l$.
Since $\cK_n$ in (\ref{eq:baysian-capital}) is always non-negative,
we have a sequential test of $H$ by Proposition \ref{prop:martingale-test}.

An advantage of the beta-binomial strategy is that the asymptotic behavior
of the capital process is easy to study by Stirling's formula.
When $n, h_n$ and $t_n$ are all large,
we can evaluate the log capital as
\[
\log \cK_n
= nD\left(\frac{h_n}{n} \Big\| \rnp\right)  
-\frac{1}{2}\log n + O(1), 
\]
where
\begin{align*}
D(p\|q) = p\log\frac{p}{q} + (1 - p)\log\frac{1 - p}{1 - q}
\end{align*}
denotes the Kullback-Leibler information between 
$0 < p < 1$ and $0 < q < 1$. 

Now we move onto the embedded coin-tossing game.
Suppose that Investor trades in a finite time interval $[0, T]$
and he uses 
the Bayesian strategy in (\ref{eq:2-5}) for the embedded coin-tossing game.
We define  $n^* = n^*(T, \delta, S(\cdot))$ by 
$t_{n^*} < T \le t_{n^*+1}$. 
Investor's capital 
${\cal K}(T) = {\cal K}^{{\cal P}^{\delta, a, b}}(T,S(\cdot))$ at $t = T$ 
for large $n^*$ is written as
 \begin{align*}
 {\cal K}(T) 
 = \tilde{{\cal K}}^*_{n^*}\left(1 + \prc_{n^*}^*
 \frac{S(T) - S(t_{n^*})}{S(t_{n^*})}\right),\quad 
 \prc_{n^*}^* = \frac{1 + \delta}{\delta (2 + \delta)}\nu_{n^*}^*.
 \end{align*} 
Since $\left|\frac{S(T) - S(t_{n^*})}{S(t_{n^*})}\right| 
< \delta$, we have
\begin{align}
\label{eq:3-1}
\log {\cal K}(T) = \log \tilde{{\cal K}}^*_{n^*} + O(1) 
= n^*D\left(\frac{h_{n^*}}{n^*} \Big\| \rnp \right)  
-\frac{1}{2}\log n^* + O(1).
\end{align} 

\subsection{Generality of high-frequency limit order strategy}

This subsection is a part which is rather independent of the
previous sections. Here 
we show a generality of the high-frequency limit order strategy
developed 
in \cite{takeuchi-kumon-takemura-bernoulli}, which implies that when
the asset price $S(t)$ follows the geometrical Brownian motion, our
strategy automatically incorporates the well-known constant proportional betting
strategy originated with Kelly (\cite{kelly}) and   
yields the likelihood ratio in the Girsanov's theorem for
geometric Brownian motion.  The convergence results in this subsection
are of measure-theoretic almost everywhere convergence.

Let $S(t)$ be subject to the geometrical Brownian motion with drift $\mu$ and volatility $\sigma$.
Then 
\begin{align*}
\log S(T)- \log S(0)
 = \bigg(\mu-\frac{1}{2}\sigma^2\bigg)T + \sigma W(T),
\end{align*}
where $W(\cdot)$ denotes the standard Brownian motion. 
In the following we write 
\[
L(T) = \log S(T)- \log S(0).
\]

We let $T\rightarrow\infty$ and 
let  $\eta=\eta_T$ depend on $T$ in such a way that $|\log \eta_T|=o(\sqrt{T})$.
Similarly we denote $\delta_T=e^{\eta_T}-1$, $\rnp_T=1/(2+\delta_T)$.
Define
\begin{align*}
&TV(\eta_T, T) = \sum_{i=1}^{n^*} |\log S(t_i) - \log S(t_{i-1})| , 
\quad 
L(\eta_T,T) = \log S(t_{n^*}) - \log S(0)  ,
\\
&  \qquad \qquad \zeta(\eta_T, T) = \frac{L(\eta_T,T)}{TV(\eta_T, T)}.
\end{align*}
We call $TV(\eta_T, T)$ the total $\eta$-variation of $\log S(t)$ in the interval $[0, T]$.
Then we have
\begin{align*}
\eta_TTV(\eta_T,T) = n^*\eta_T^2 = \sigma^2T + O(\eta_T),
\end{align*}
and hence we can evaluate
\begin{align*}
\theta_{n^*}^* = \frac{\mu}{\sigma^2}+ \frac{W(T)}{\sigma T} + O(\eta_T).
\end{align*}
Note that when $L(T)$ is a process symmetric around the origin with 
$\mu - \frac{1}{2}\sigma^2 = 0$, we have
\begin{align*}
\theta_{n^*}^* = \frac{1}{2}+ \frac{W(T)}{\sigma T} + O(\eta_T),
\end{align*}
and the main term\ $1/2$ in the right-hand side indicates the even
rebalanced strategy between the asset and the cash. 

Let us consider $n^* D(p(\eta_T,T) \| \rnp_T)$,
where $n^* = TV(\eta_t, T)/\eta_T$, $p(\eta_T, T)=h_{n^*}/n^* 
=(1 + \zeta(\eta_T,T))/2$.
{}From the Taylor expansion
\begin{align*}
D\left(\frac{1 + d_1}{2}\Big\| \frac{1 + d_2}{2}\right) = 
\frac{1}{2}(d_1 - d_2)^2 + O(|d_1 - d_2|^3),
\end{align*}
with $d_1 = \zeta(\eta_T, T),\ d_2 = - \delta_T/2$, we have 
\begin{align*}
n^* D\left(p(\eta_T, T) \| \rnp_T \right) &= 
\frac{n^*}{2}\bigg(\zeta(\eta_T,T)+\frac{\delta_T}{2}\bigg)^2+O(\eta_T^3)
= \frac{n^*\eta_T^2}{2}(\theta_{n^*}^*)^2+O(\eta_T^3)\\
&= \frac{\sigma^2T}{2}\bigg(\frac{\mu}{\sigma^2}+\frac{W(T)}{\sigma T}\bigg)^2 = \frac{\mu W(T)}{\sigma} + \frac{\mu^2 T}{2\sigma^2} + 
\frac{W^2(T)}{2T} + O(\eta_T^3).
\end{align*}
The log capital
$\log {\cal K}(T) = n^*D\left(p(\eta_T, T) \| \rnp_T\right)  -\frac{1}{2}\log n^* + O(1)$
is expressed as
\begin{align*}
\log {\cal K}(T) = \frac{\mu W(T)}{\sigma} + \frac{\mu^2 T}{2\sigma^2} -\frac{1}{2}\log T + \log \eta_T + O(1),
\end{align*}
and hence when $|\log \eta_T| = o(\sqrt{T})$, the main terms in the right-hand side of $-\log {\cal K}(T)$,
\begin{align*}
-\log {\cal K}(T) = -\frac{\mu W(T)}{\sigma} - \frac{\mu^2 T}{2\sigma^2} 
 + o(\sqrt{T})
\end{align*}
provide the likelihood ratio of the unique martingale measure known as the Girsanov's theorem, and we obtain
\begin{align*}
\lim_{T\to \infty}\frac{\log {\cal K}(T)}{T} = \frac{\mu^2}{2\sigma^2}.
\end{align*}

\subsection{Markov betting strategy}
\label{subsec:Markov}
The Bayesian strategy in previous subsections is a simple strategy  based on
the past number of ups and downs only.  The strategy does not exploit
possible autocorrelations in the ups and downs of the price process.
Multistep Bayesian strategy, in particular the Markov type strategy 
in \cite{takeuchi-kumon-takemura-markov} is very efficient in
exploiting possible autocorrelations.  In this paper we just use the first-order  Markov strategy.

For the biased-coin game the strategy is given as
\[
\nu_1 = 0,\quad
\nu_n = 
\begin{cases} \nu_n^+ , &  \textrm{if}\ x_{n-1} = 1 \\
              \nu_n^-, & \textrm{if}\ x_{n-1} = 0
\end{cases}
\quad n = 2, 3, \dots,
\]
where $\nu_n^+$ and $\nu_n^-$ can have different values.  It 
incorporates the information on the last move $x_{n-1}$ of Market.

We use the beta-binomial strategy separately for 
the case of $x_{n-1} = 1$  and $x_{n-1}=0$ with hyperparameters $a,b$.
Let $q_n^1=h_n$ and $q_n^0=n-h_n$.
Denote the numbers of pairs $(x_{i-1} x_i) = (1
1),(10), (01), (00)$, $i = 2, \dots, n$, by $q_n^{11},\  q_n^{10},\
q_n^{01},\ q_n^{00}$, respectively.
The capital $\cK_n=\cK_n^{\cP_Q}$  for this strategy is given by
\begin{align*}
\cK_n 
&= \frac{\Gamma(a + b)^2\Gamma(q_n^{11} + a)
\Gamma(q_n^{10} + b)\Gamma(q_n^{01} + a)\Gamma(q_n^{00} + b)}
{\Gamma(a)^2\Gamma(b)^2\Gamma(q_n^{11} + q_n^{10} + a + b)
\Gamma(q_n^{01} + q_n^{00} + a + b)
\rnp^{q_n^{11}+q_n^{01}}(1 - \rnp)^{q_n^{10}+q_n^{00}}}.
\end{align*}
By Stirling's formula the asymptotic behavior of $\cK_n$ is easily derived.

We can apply the above first-order Markov strategy to the embedded coin-tossing
game for continuous price paths.  In 
\cite{takeuchi-kumon-takemura-markov} the performance of the strategy
for small grid size $\delta$ is analyzed as follows.  Under some regularity
conditions, if the price process path has the H\"older exponent $H\neq 1/2$, then
\begin{align}
\label{eq:2-5-1}
\log  \cK_{n^*}  = n^* D\Big(\frac{1}{2^{1/H-1}} \Big\| \frac{1}{2} \Big) + o(n^*).
\end{align}

\section{Tests of martingale property}
\label{sec:martingale}

In Section \ref{subsec:limit-order} we have already shown that
we can test $\bar H$ in (\ref{eq:continuous-martingale}) by limit order type
betting strategy.  It is an important fact that the converse is also true.
If the null hypothesis
$H$ in (\ref{eq:null-hypothesis}) holds for embedded discrete time coin-tossing
game for every $\delta>0$, then $\bar H$ holds.  This fact can be proved
by adapting the result of Dubins and Schwartz 
\cite{dubins-schwarz} to positive  martingales.
A more rigorous and modern treatment of the result of Dubins and
Schwartz is given in Chapter V of \cite{revuz-yor}.  Recently
Vovk \cite{vovk0904.4364} gave a complete generalization of these
results to game-theoretic framework.
However for our purposes, the arguments given in \cite{dubins-schwarz}
are sufficient and more suitable, because they are 
based on similar ideas to our limit order type strategies.

\begin{proposition}
\label{prop:equivalence}
Let $S(t)$, $t\ge 0$, 
be a continuous positive stochastic process with $S(0)=1$ such that
almost all of whose paths are nowhere constant 
and $\limsup_t S(t)=\infty$, $\liminf_t S(t)=0$.
Then the following three conditions are equivalent.
\begin{enumerate}
\setlength{\itemsep}{-2pt}
\item $S$ is a martingale.
\item $S$ is a path-dependent and future-independent time change of the
standard geometric Brownian motion.
\item For every $\eta>0$,  
the directions  $x_j$, $j=1,2,\dots$, in (\ref{eq:def-xn}) are independently and identically 
distributed  with
\[
P(x_j=1)= \frac{1}{1+e^{\eta}}
\]
and they are independent of the waiting times  $\{ w_j, j=1,2,\dots\}$.
\end{enumerate}
\end{proposition}

\begin{proof}
Since our proof is a simple adaptation of the proof in
\cite{dubins-schwarz}  we only give an outline of the proof.
The implications 2 $\Rightarrow$ 1 and 1 $\Rightarrow$ 3 are obvious.
Therefore it suffices to prove 3 $\Rightarrow$ 2.  By examining the
proof in \cite{dubins-schwarz}, we note that 
the only difference in our proof 
is the finite dimensional distributions of $\log S(t)$ at the
trading times $t_1, t_2,\dots$. For simplicity we consider the distribution
of $\log S(t_2)-\log S(t_1)$ for any fixed $\eta$.  Now consider a
sequence of increasingly finer horizontal
grids $\eta_m = 2^{-m}\eta$ in the logarithmic scale for $\log S$. Then
$\log S(t_2)-\log S(t_1)$ ca be written 
\[
\log S(t_2)-\log S(t_1) = \sum_{j=1}^{2^m} \epsilon_{mj}\eta_m,
\]
where $\epsilon_{mj}$, $j=1,\dots,2^m$, are i.i.d.\ $\pm 1$
random variables with $P(\epsilon_{ml}=1)=1/(1+e^{\eta_m})$. 
The mean and the variance of $\sum_{j=1}^{2^m} \epsilon_{mj}\eta_m$
are given  by
\[
2^m (2p_{mj}-1), 
 \quad 4 \times 2^m p_{mj}(1-p_{mj}), \qquad
 \big(\ p_{mj}=\frac{1}{1+e^{\eta_m}} \ \big).
\]
As $m\rightarrow\infty$,  these converge to
$-\eta/2$ and $\eta$, respectively.  Also by the central limit theorem
$\log S(t_2)-\log S(t_1)$ is distributed according to
$N(-\eta/2, \eta)$. Considering any finite number 
trading times, we see that finite-dimensional distributions
are the same as the geometric Brownian motion.
Now an argument similar to \cite{dubins-schwarz}
shows that $S$ satisfies 2.
\end{proof}

Because of  Proposition \ref{prop:equivalence} it is natural
to test the martingale hypothesis
(\ref{eq:continuous-martingale}) by testing
that the directions  $x_j$, $j=1,2,\dots$, are i.i.d.\ \{0, 1\} valued  
random variables with the probability $P(x_j=1)=1/(1+e^{\eta})$.
Note that we can interpret the distribution of the waiting
times as nuisance parameters of the null hypothesis.

In the next section we use two strategies and associated sequential tests.
The first strategy is  a simple Bayesian strategy of 
Section \ref{subsec:bayesian}
concerning the success probability $P(x_j=1)=1/(1+e^{\eta})$. The
second strategy is the first-order Markov strategy  of Section 
\ref{subsec:Markov} for testing independence.

\section{Numerical examples}
\label{sec:numerical}

In this section we give some numerical examples on the stock price
data from the Tokyo Stock Exchange. The data are the stock minute
prices from June 1st to July 31st in 2006 for three Japanese companies
SoftBank, IHI, and Sony listed on the first section of the TSE, which were adapted from Bloomberg LP. Usually there are 270 minute price data a day.

We employed a simple Bayesian strategy and a Markov strategy.  The former
did not show significant results.  This is because
the empirical probabilities  of heads ($p^1({\rm ft})$ in Table 1) 
are close to $0.5$.
However we obtained
significant
results by the first-order Markov strategy for three companies 
with common values $\eta = 2^{-k},\ a, b = 0.01\cdot 2^k,\ k = 8$.
The results are shown in Figures 1--15.  Figures 1--5 are for SoftBank,
Figures 6--10 are for IHI and Figures 11--15 are for Sony.
In each figure, fn denotes the first round such that the Markov capital MK satisfies $\textrm{MK(fn)} \ge 10^3$, and ft denotes the approximate time of fn in minutes. 
By Proposition \ref{prop:martingale-test}, $10^3$ corresponds
to the significance level of $\alpha=0.1\%$.  

We also exhibit processes of empirical probabilities $p^{1|1}, p^{0|0}, p^1$, and processes of H\"older exponents $H^1, H^0$, which are given in the following manner.
\begin{align}
\label{eq:5-1}
&p^{1|1}_n = \frac{q^{11}_n}{q^1_n},\quad 
p^{0|0}_n = \frac{q^{00}_n}{q^0_n},\quad  
p^1_n = \frac{q^1_n}{n},\nonumber \\
&\frac{1}{2^{1/H^1_n} - 1} = p^{1|1}_n,\quad 
\frac{1}{2^{1/H^0_n} - 1} = p^{0|0}_n.
\end{align}
The relation (\ref{eq:5-1}) between the conditional probability and
the H\"older exponent is one of the results obtained in
\cite{takeuchi-kumon-takemura-markov}. These typical values are summarized in 
Table 1.

\begin{figure}[htbp]
\begin{minipage}{.5\linewidth}
\begin{flushleft}
Figure 1 : Minute prices of SoftBank\\
Figure 2 : Capital process of Markov\\
\qquad \qquad \quad strategy\\
Figure 3 : Log capital process of Markov\\
\qquad \qquad \quad strategy\\
Figure 4 : Processes of empirical \\
\qquad \qquad \quad probabilities $p^{1|1},\ p^{0|0},\ p^1$\\
Figure 5 : Processes of H\"older exponents \\
\qquad \qquad \quad $H^1, H^0$
\end{flushleft}
\end{minipage}
\begin{minipage}{.5\linewidth}
\begin{center}
\includegraphics[width=8cm,height=7cm]{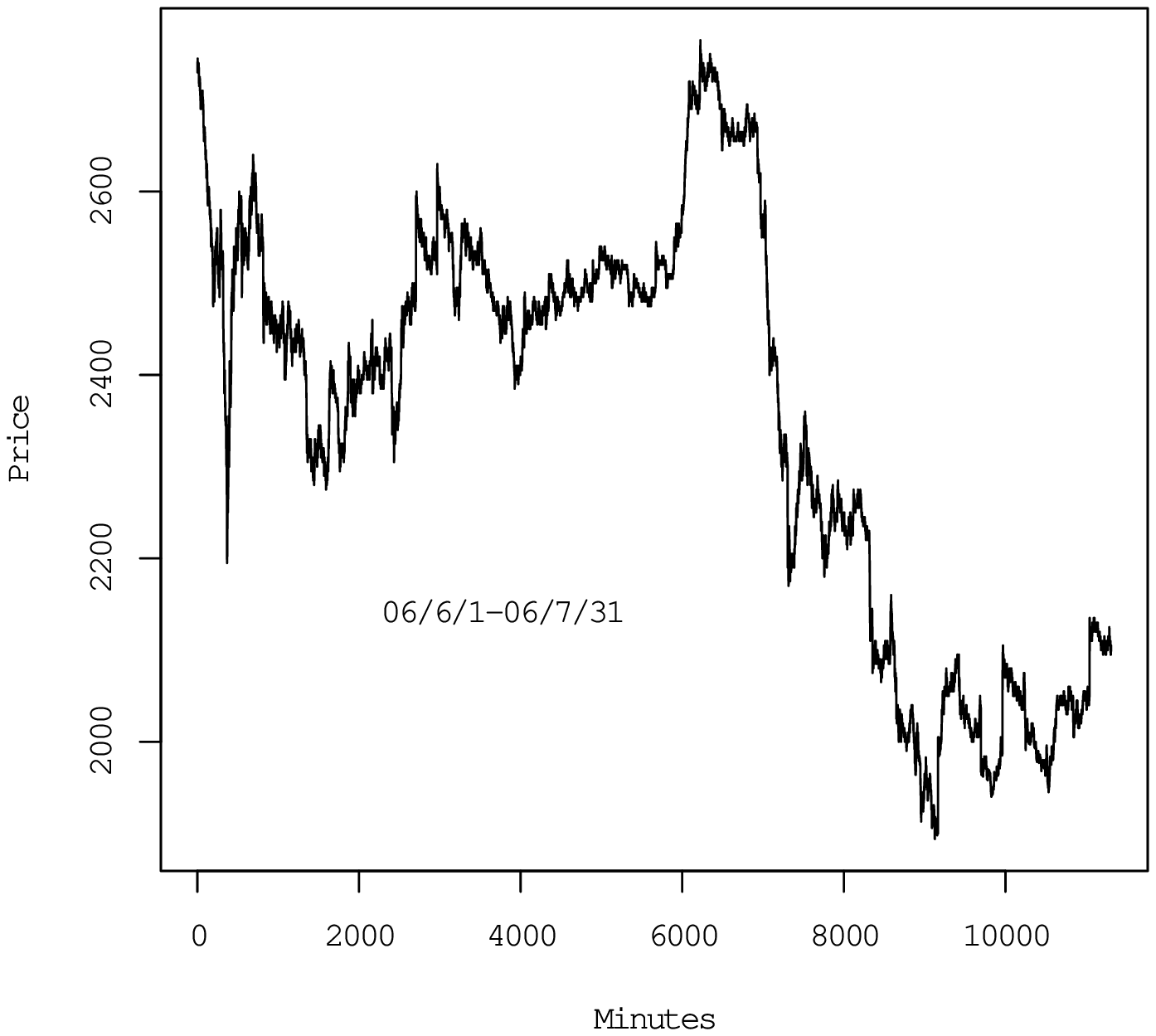}
\vspace*{-11mm}   
\caption{SoftBank minute prices}
\label{fig:1}
\end{center}
\end{minipage}
\begin{minipage}{.5\linewidth}
\begin{center}
\includegraphics[width=8cm,height=7cm]{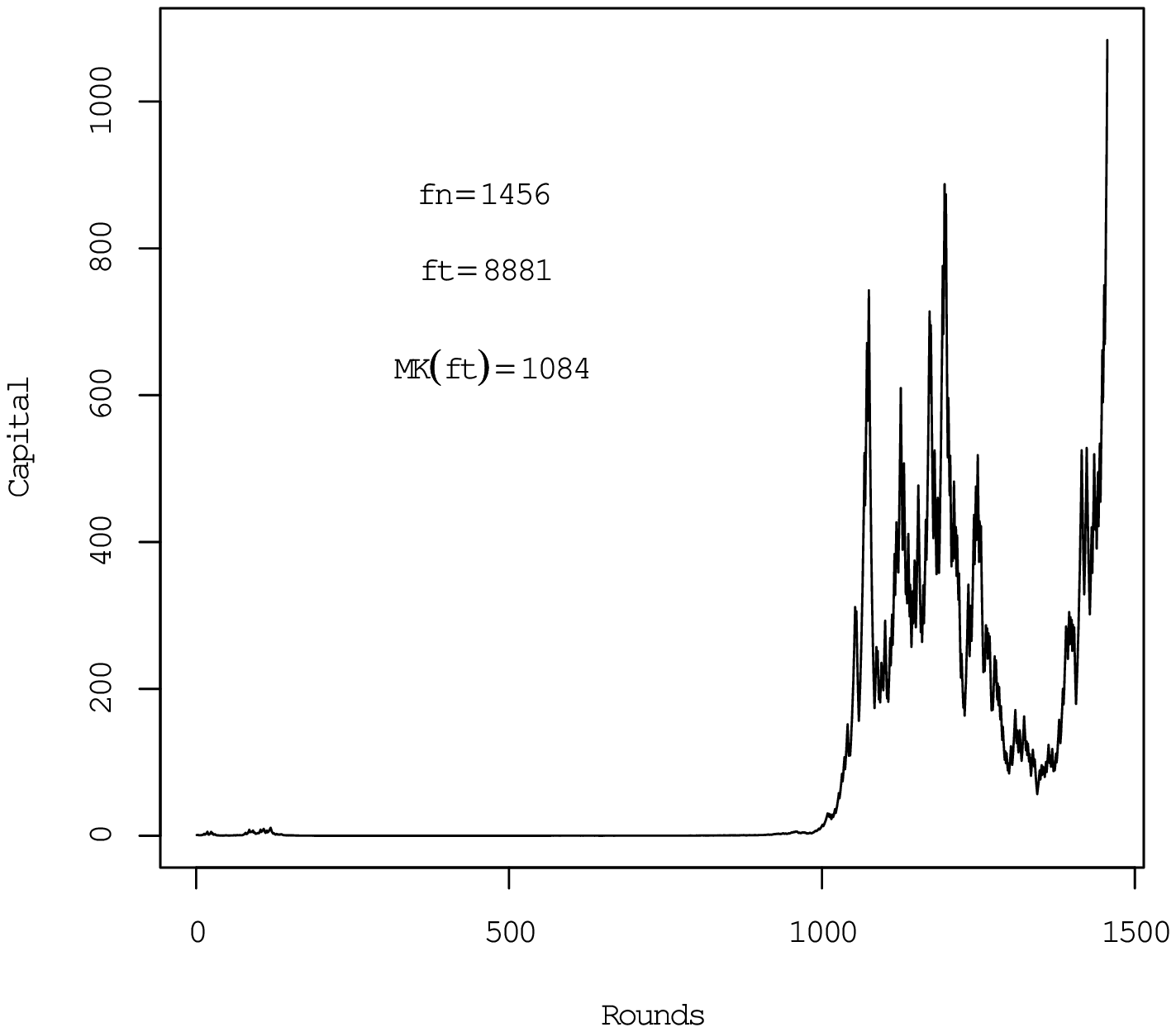}
\vspace*{-11mm}   
\caption{Markov capital process}
\label{fig:2}
\end{center}
\end{minipage}
\begin{minipage}{.5\linewidth}
\begin{center}
\includegraphics[width=8cm,height=7cm]{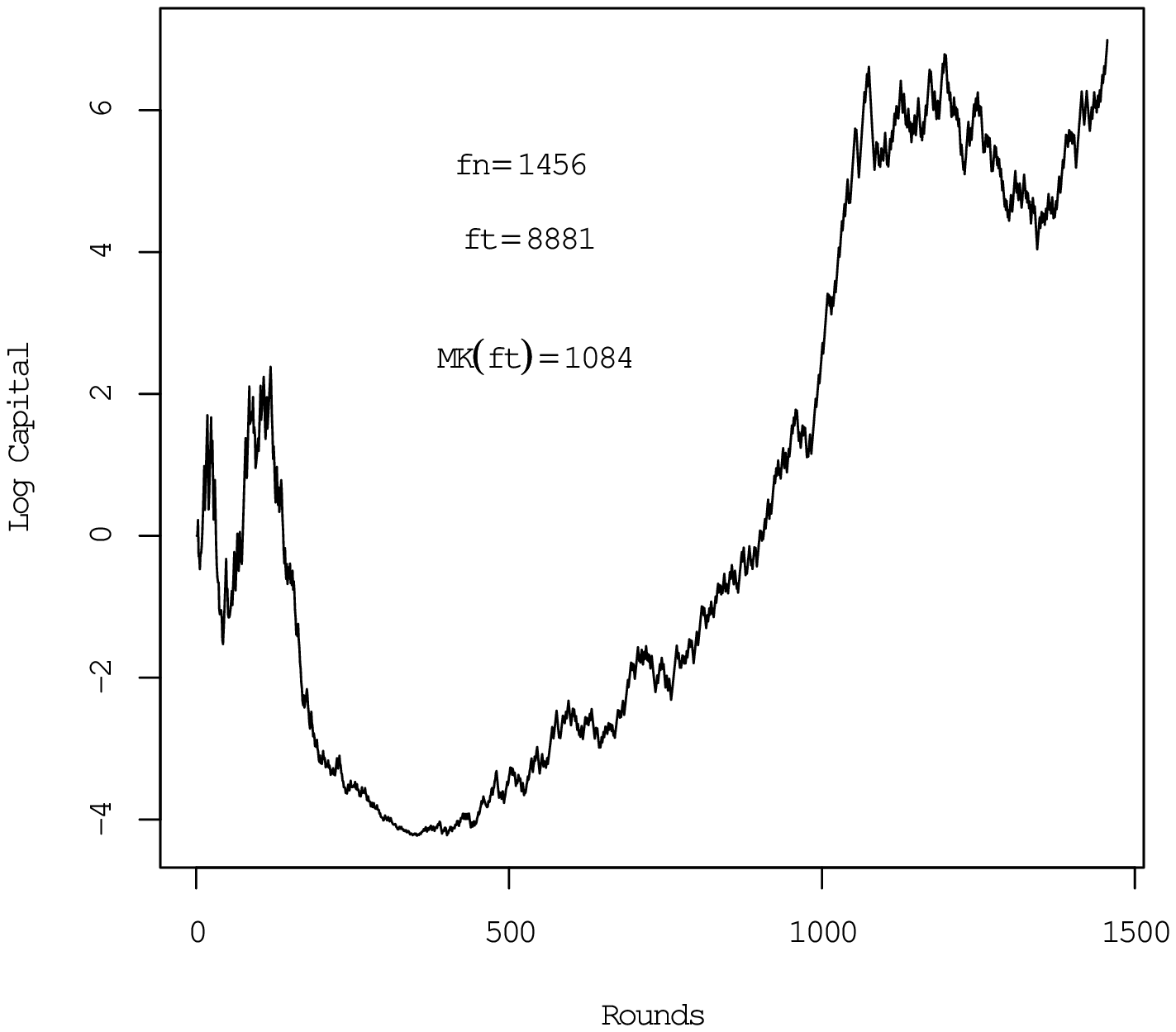}
\vspace*{-11mm}   
\caption{Log Markov capital process}
\label{fig:3}
\end{center}
\end{minipage}
\begin{minipage}{.5\linewidth}
\begin{center}
\includegraphics[width=8cm,height=7cm]{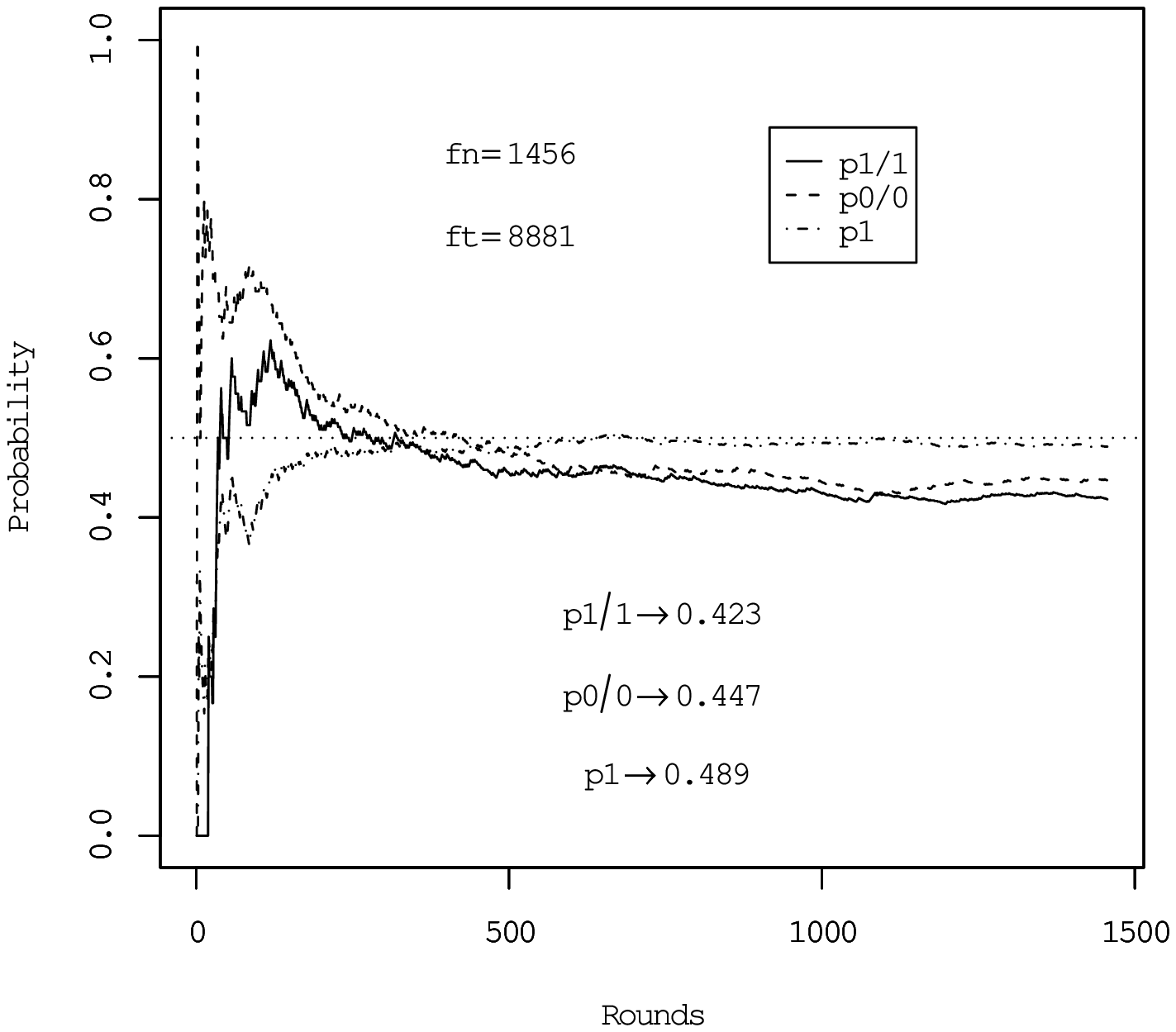}
\vspace*{-11mm}   
\caption{Empirical probability processes}
\label{fig:4}
\end{center}
\end{minipage}
\begin{minipage}{.5\linewidth}
\begin{center}
\includegraphics[width=8cm,height=7cm]{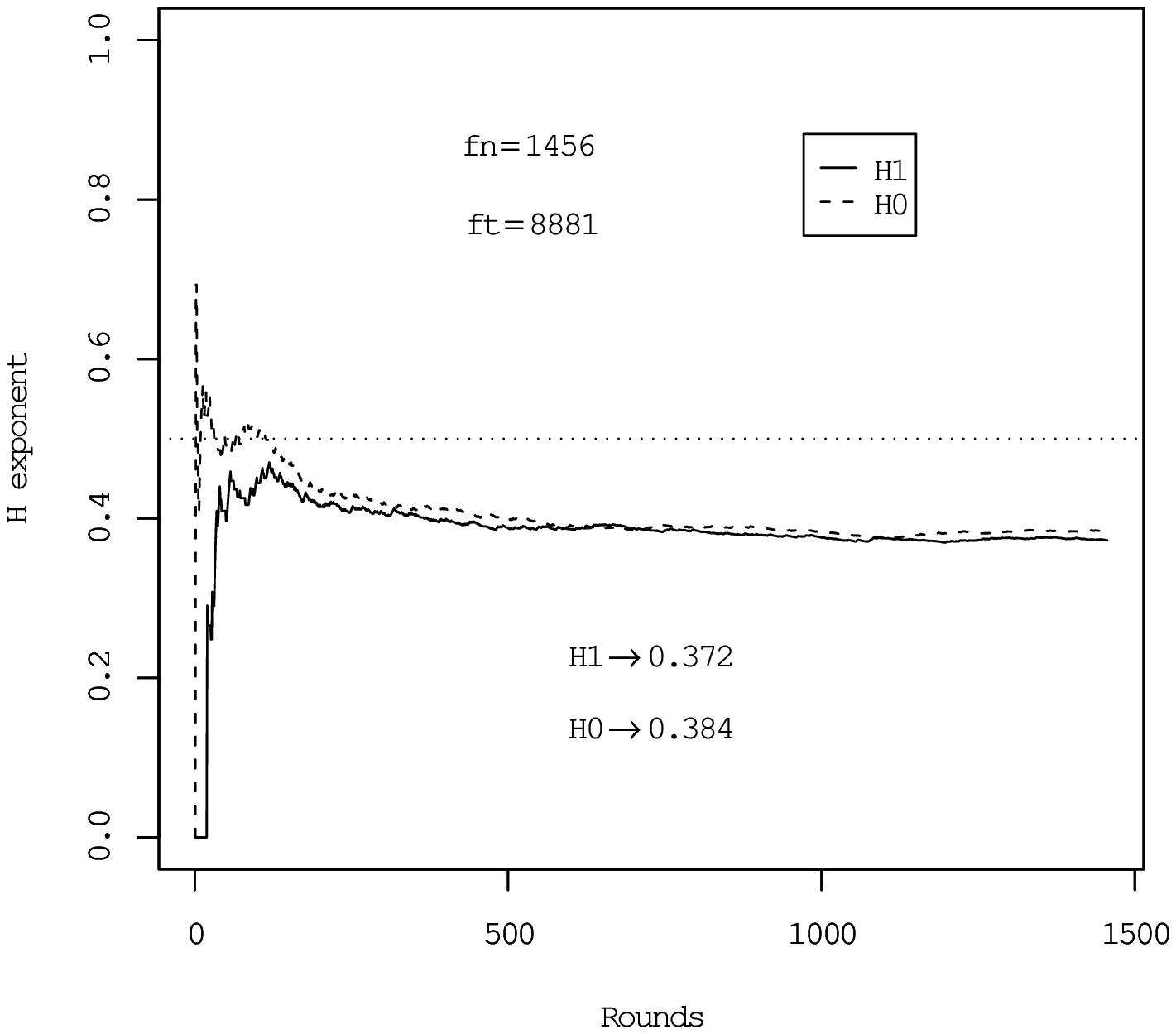}
\vspace*{-11mm}   
\caption{H\"older exponent processes}
\label{fig:5}
\end{center}
\end{minipage}
\end{figure}

\begin{figure}[htbp]
\begin{minipage}{.5\linewidth}
\begin{flushleft}
Figure 6 : Minute prices of IHI\\
Figure 7 : Capital process of Markov\\
\qquad \qquad \quad strategy\\
Figure 8 : Log capital process of Markov\\
\qquad \qquad \quad strategy\\
Figure 9 : Processes of empirical \\
\qquad \qquad \quad probabilities $p^{1|1},\ p^{0|0},\ p^1$\\
Figure 10 : Processes of H\"older exponents \\
\qquad \qquad \quad \ $H^1, H^0$
\end{flushleft}
\end{minipage}
\begin{minipage}{.5\linewidth}
\begin{center}
\includegraphics[width=8cm,height=7cm]{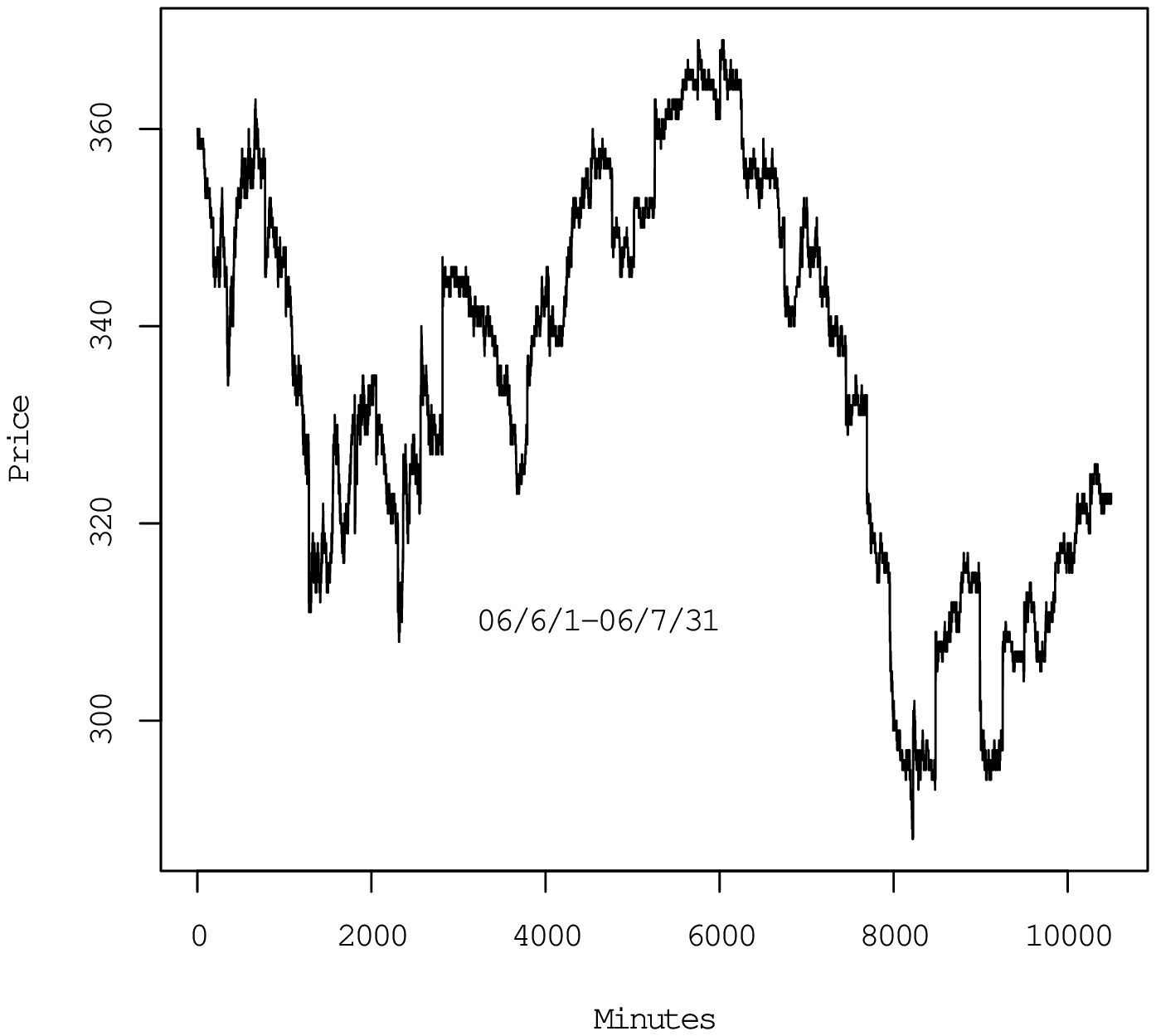}
\vspace*{-11mm}   
\caption{IHI minute prices}
\label{fig:6}
\end{center}
\end{minipage}
\begin{minipage}{.5\linewidth}
\begin{center}
\includegraphics[width=8cm,height=7cm]{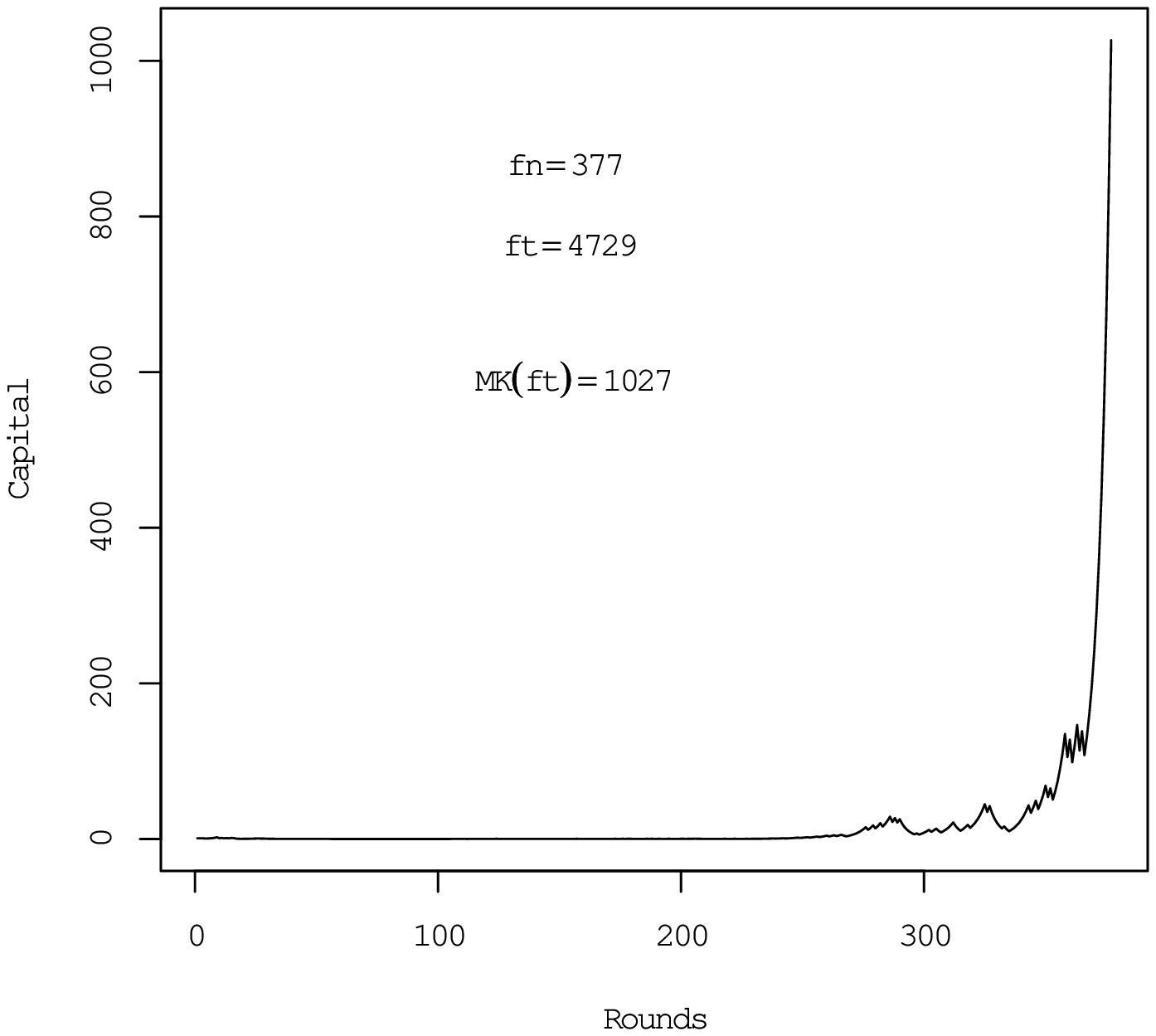}
\vspace*{-11mm}   
\caption{Markov capital process}
\label{fig:7}
\end{center}
\end{minipage}
\begin{minipage}{.5\linewidth}
\begin{center}
\includegraphics[width=8cm,height=7cm]{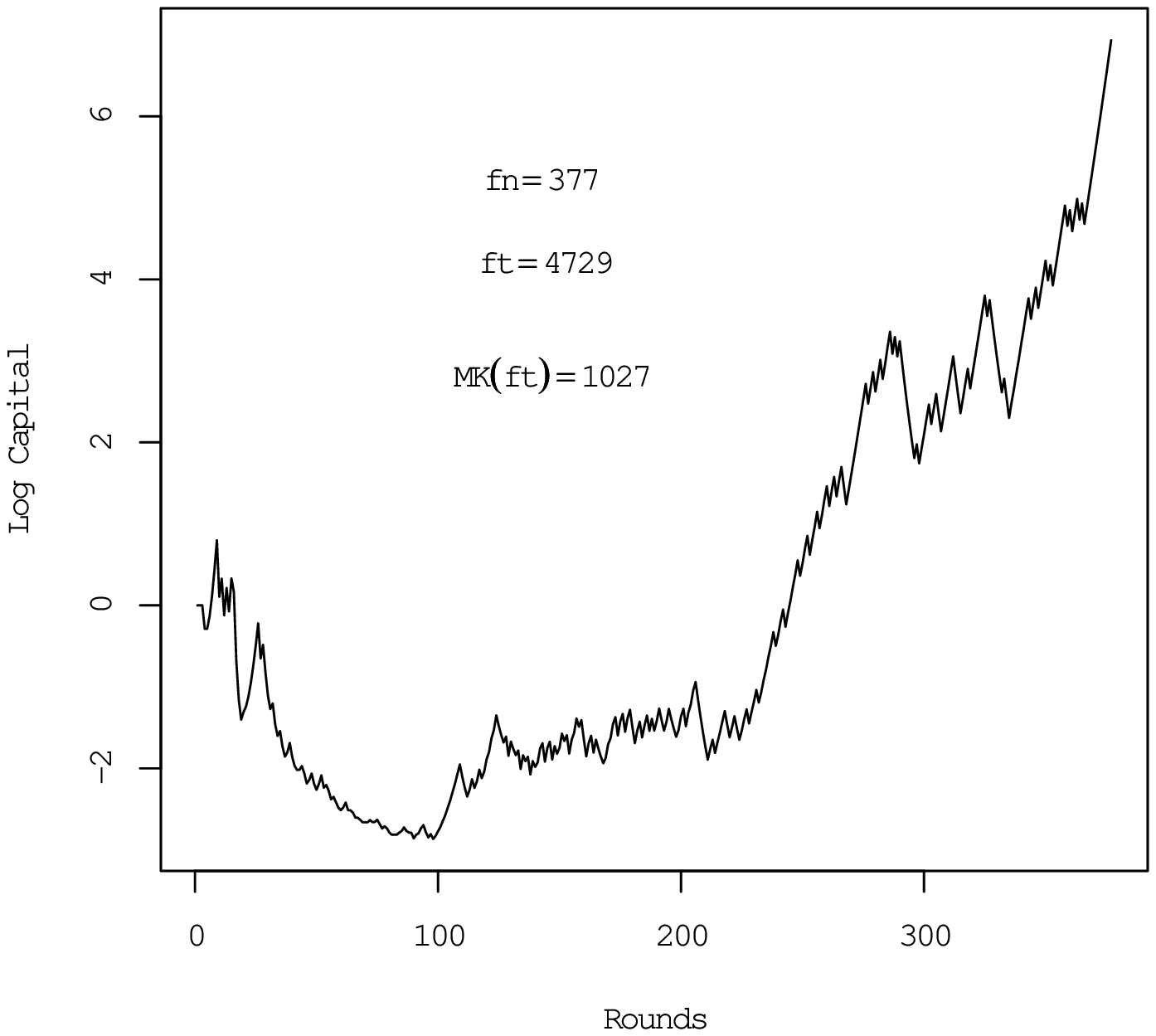}
\vspace*{-11mm}   
\caption{Log Markov capital process}
\label{fig:8}
\end{center}
\end{minipage}
\begin{minipage}{.5\linewidth}
\begin{center}
\includegraphics[width=8cm,height=7cm]{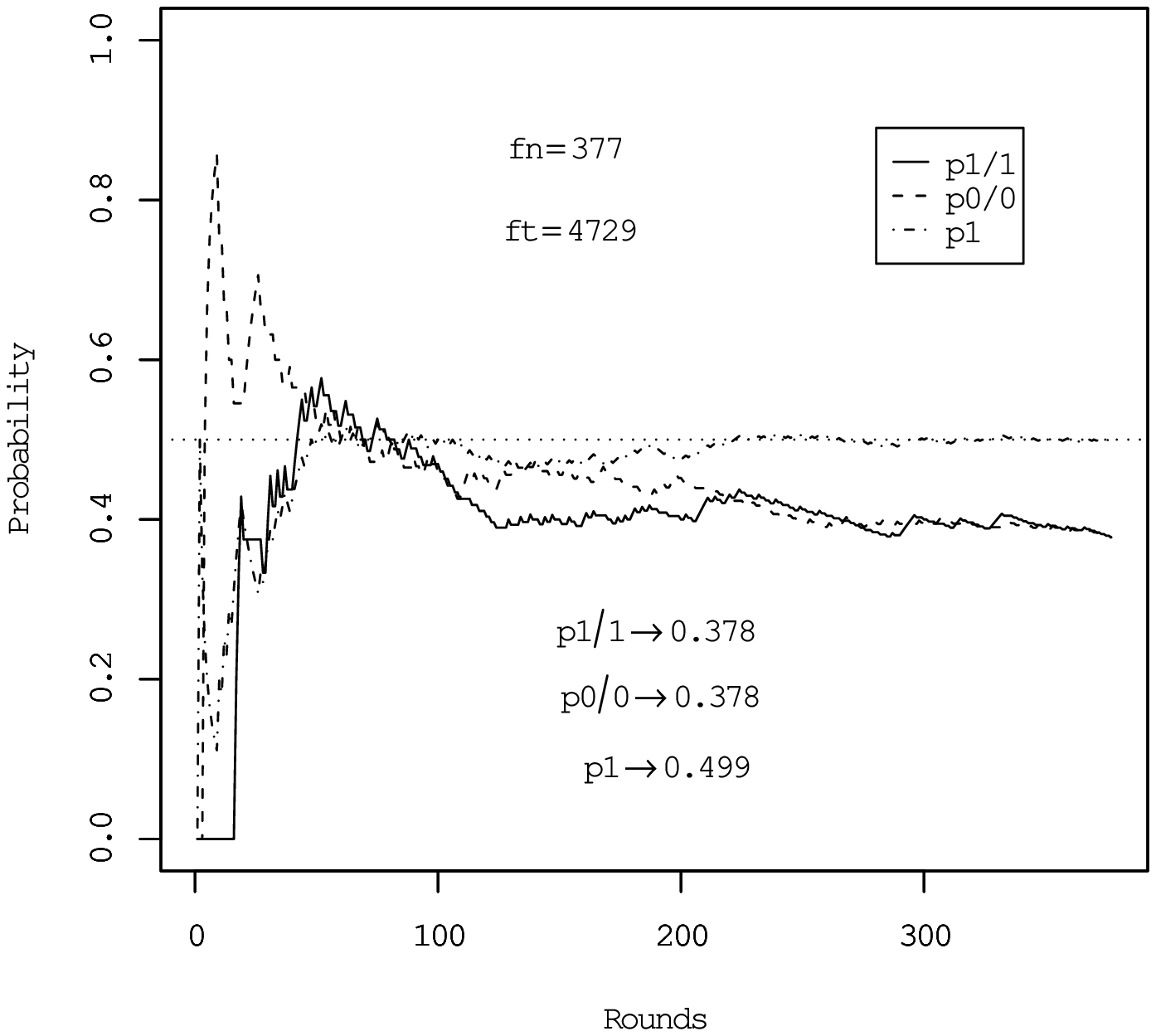}
\vspace*{-11mm}   
\caption{Empirical probability processes}
\label{fig:9}
\end{center}
\end{minipage}
\begin{minipage}{.5\linewidth}
\begin{center}
\includegraphics[width=8cm,height=7cm]{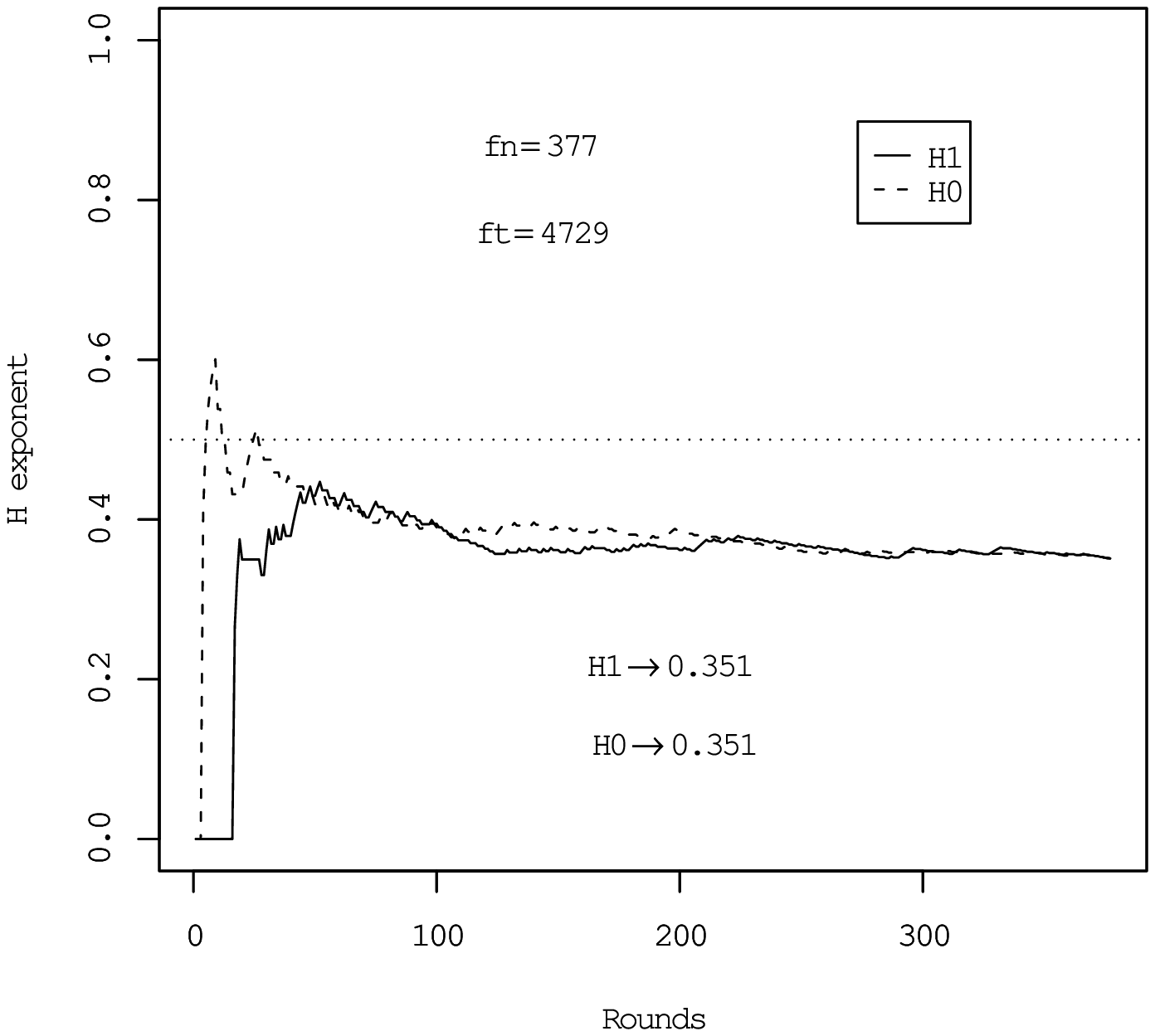}
\vspace*{-11mm}   
\caption{H\"older exponent processes}
\label{fig:10}
\end{center}
\end{minipage}
\end{figure}

\begin{figure}[htbp]
\begin{minipage}{.5\linewidth}
\begin{flushleft}
Figure 11 : Minute prices of Sony\\
Figure 12 : Capital process of Markov\\
\qquad \qquad \quad \ strategy\\
Figure 13 : Log capital process of Markov\\
\qquad \qquad \quad \ strategy\\
Figure 14 : Processes of empirical \\
\qquad \qquad \quad \ probabilities $p^{1|1},\ p^{0|0},\ p^1$\\
Figure 15 : Processes of H\"older exponents \\
\qquad \qquad \quad \ $H^1, H^0$
\end{flushleft}
\end{minipage}
\begin{minipage}{.5\linewidth}
\begin{center}
\includegraphics[width=8cm,height=7cm]{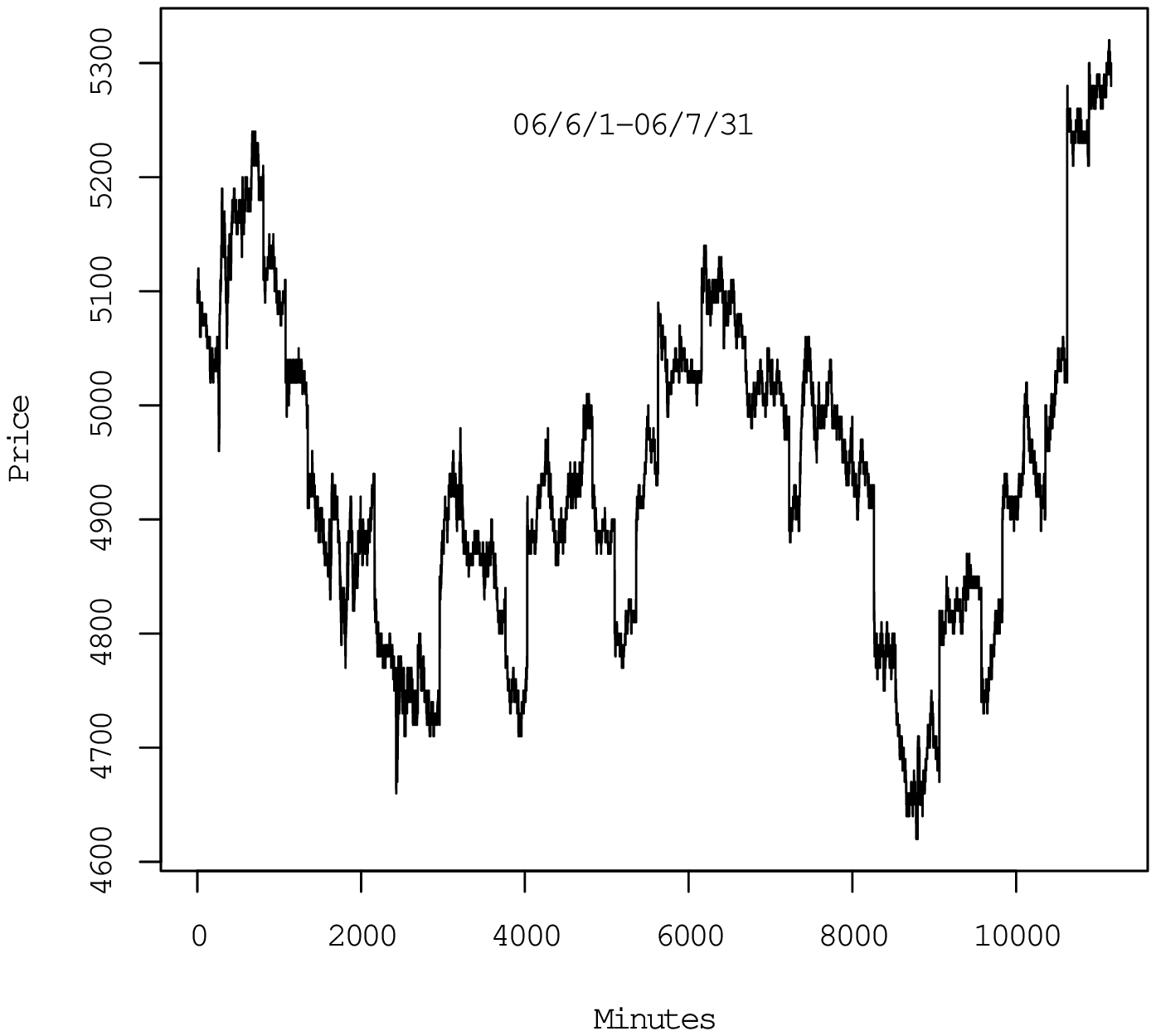}
\vspace*{-11mm}   
\caption{Sony minute prices}
\label{fig:11}
\end{center}
\end{minipage}
\begin{minipage}{.5\linewidth}
\begin{center}
\includegraphics[width=8cm,height=7cm]{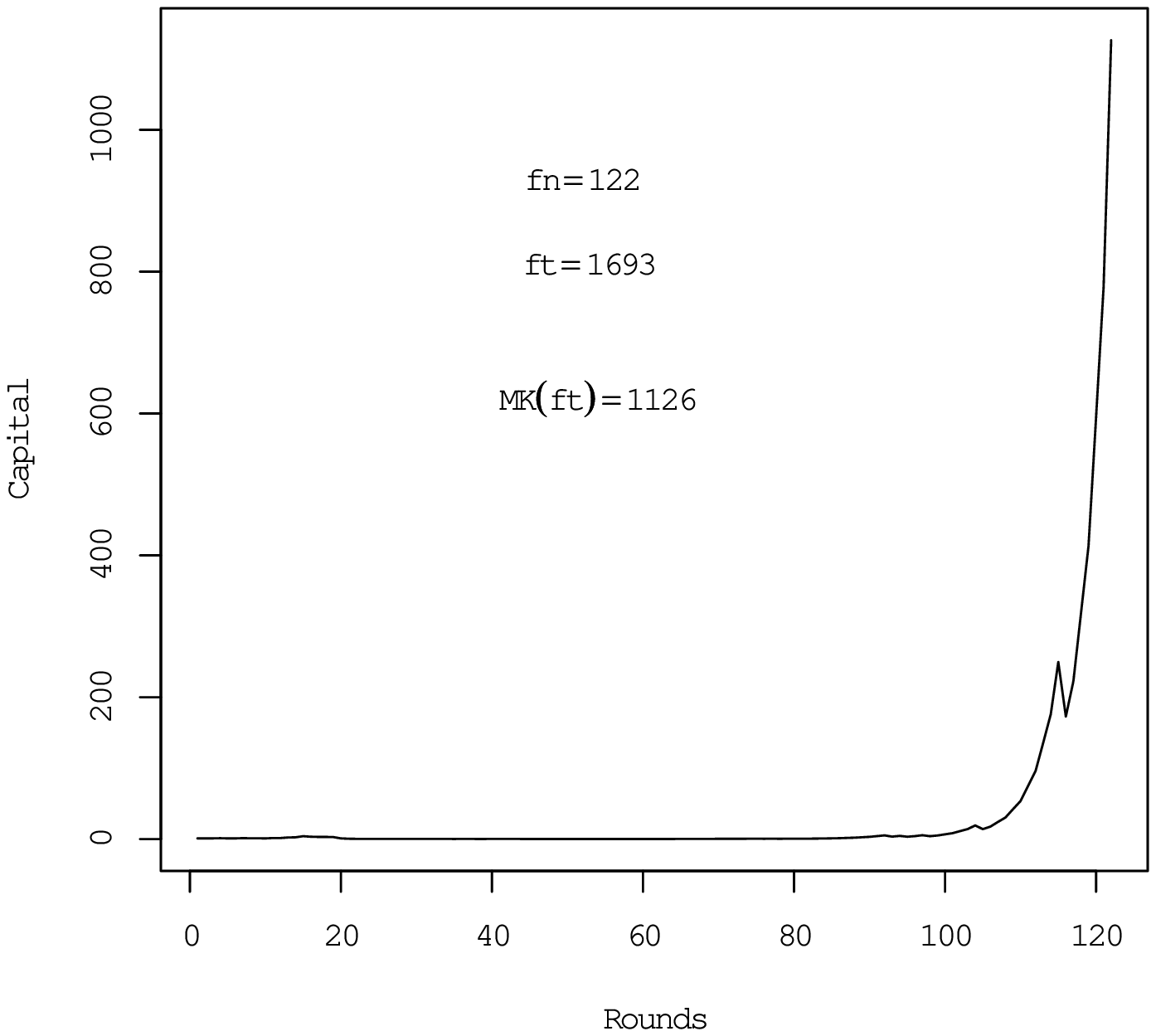}
\vspace*{-11mm}   
\caption{Markov capital process}
\label{fig:12}
\end{center}
\end{minipage}
\begin{minipage}{.5\linewidth}
\begin{center}
\includegraphics[width=8cm,height=7cm]{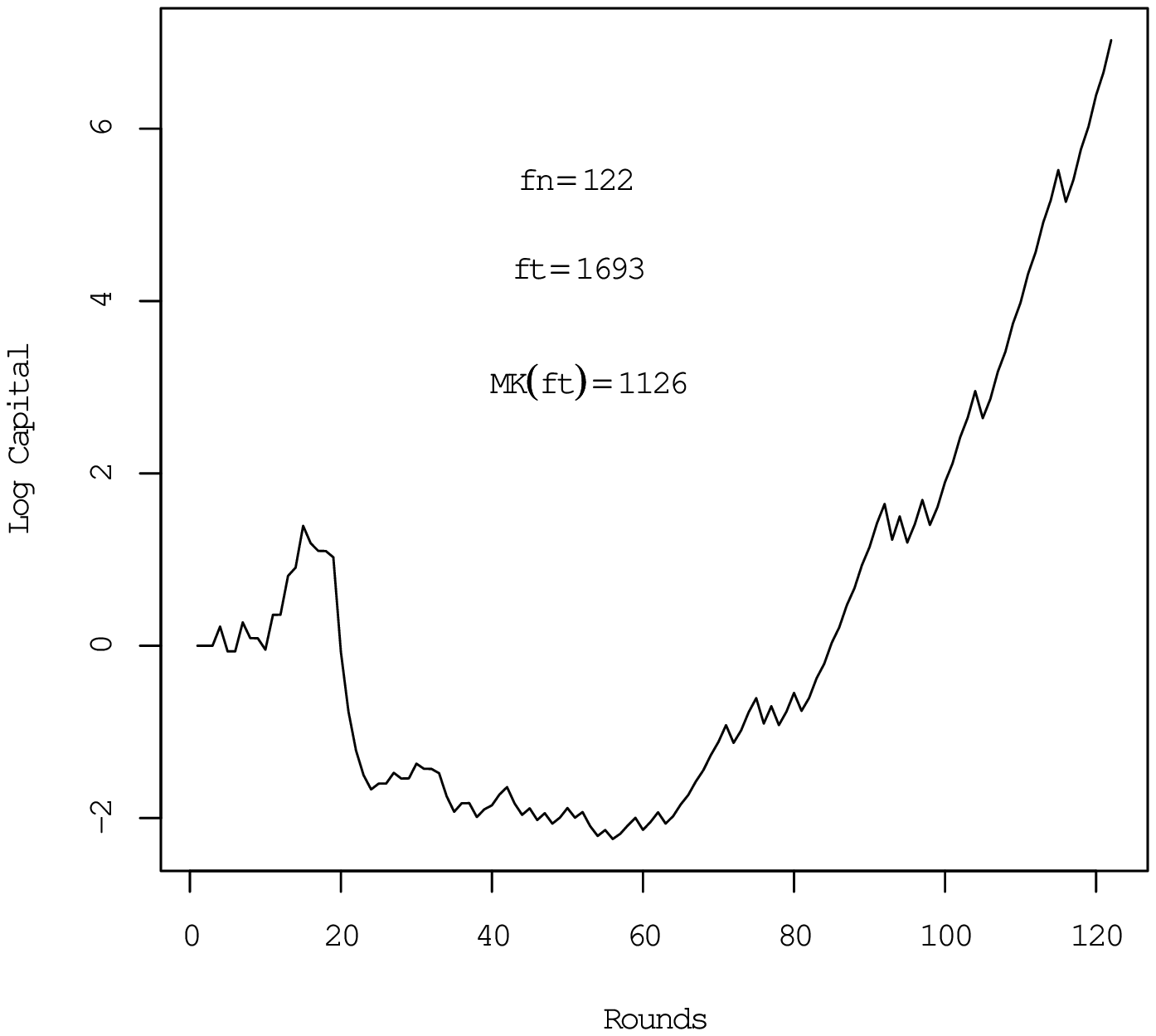}
\vspace*{-11mm}   
\caption{Log Markov capital process}
\label{fig:13}
\end{center}
\end{minipage}
\begin{minipage}{.5\linewidth}
\begin{center}
\includegraphics[width=8cm,height=7cm]{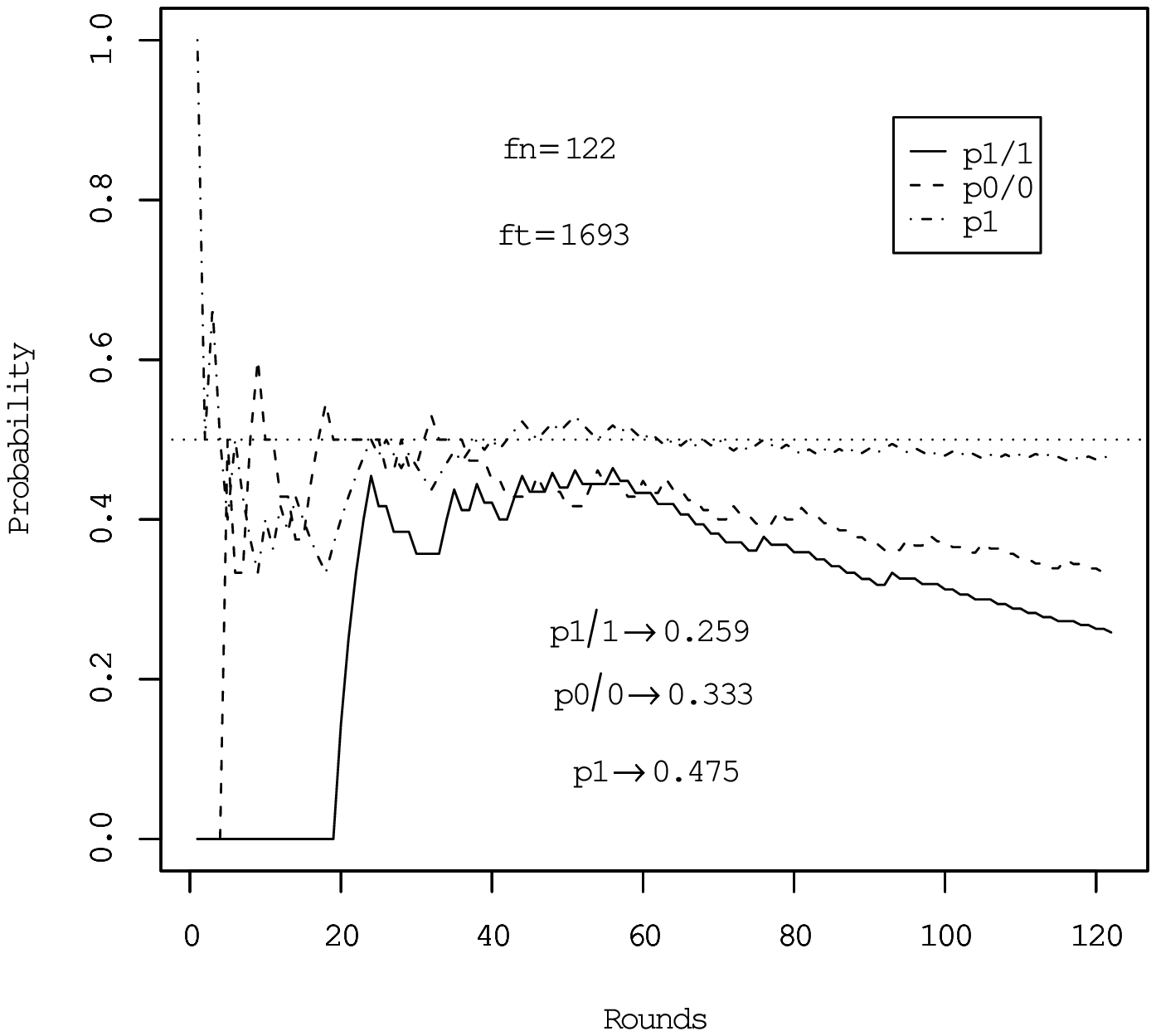}
\vspace*{-11mm}   
\caption{Empirical probability processes}
\label{fig:14}
\end{center}
\end{minipage}
\begin{minipage}{.5\linewidth}
\begin{center}
\includegraphics[width=8cm,height=7cm]{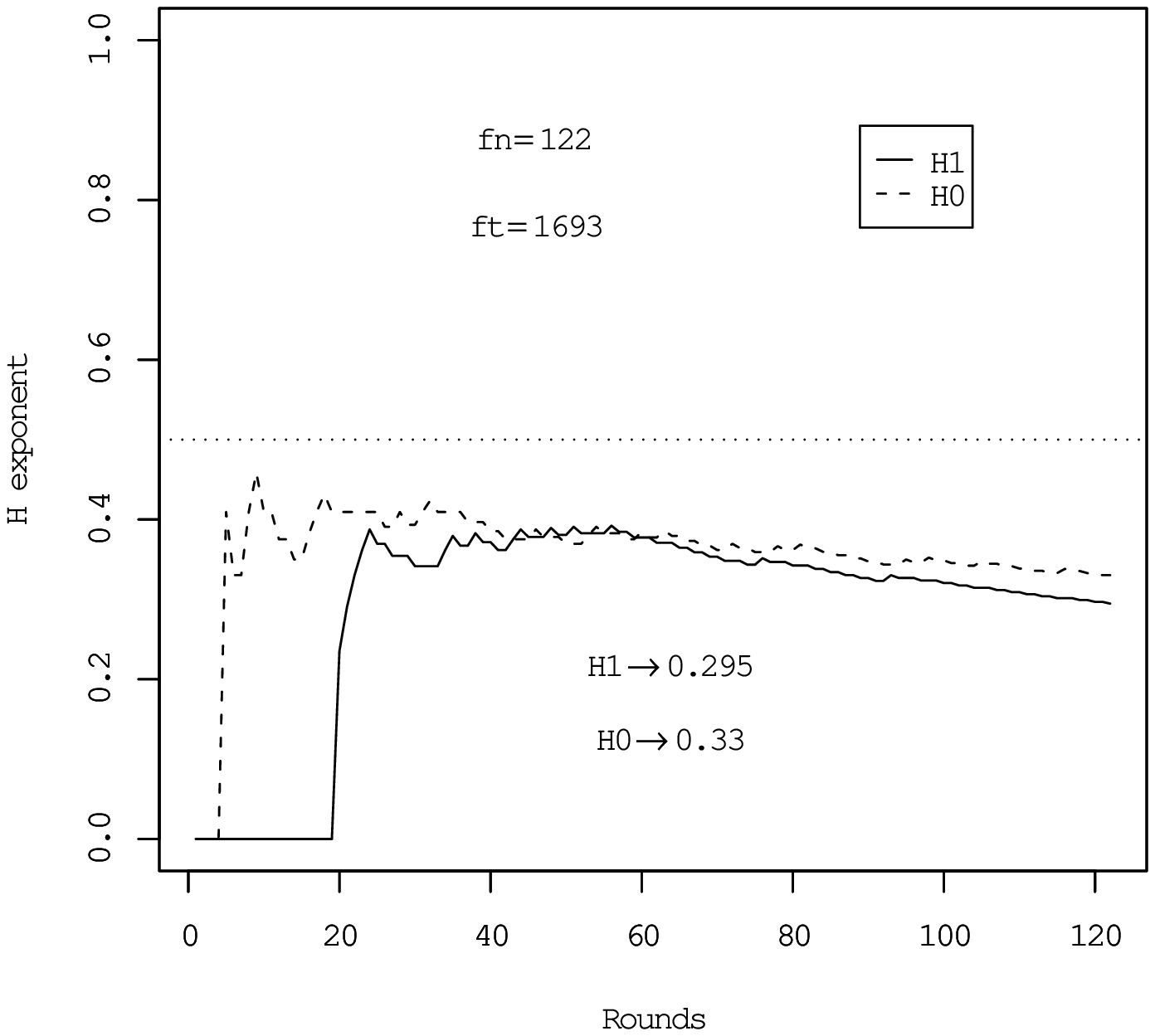}
\vspace*{-11mm}   
\caption{H\"older exponent processes}
\label{fig:15}
\end{center}
\end{minipage}
\end{figure}

\begin{table}[htbp]
\begin{center}
\label{tab:1}
\caption{Typical values for three stock minute prices}
\ 

\begin{tabular}{c|rrrrrrrr}
\hline
& fn & ft & MK(ft) & $p^{1|1}$(ft) & $p^{0|0}$(ft) & $p^1$(ft) & $H^1$(ft) & $H^0$(ft)\\
\hline
SoftBank & 1456 & 8881 & 1084 & 0.423 & 0.449 & 0.489 & 0.372 & 0.384\\
IHI & 377 & 4729 & 1027 & 0.378 & 0.378 & 0.499 & 0.351 & 0.351\\
Sony & 122 & 1693 & 1126 & 0.259 & 0.333 & 0.475 & 0.295 & 0.330\\
\hline
\end{tabular}
\end{center}
\end{table}

\begin{figure}[htbp]
\begin{minipage}{.5\linewidth}
\begin{flushleft}
Figure 16 : Log capital processes of Markov\\
\qquad \qquad \quad strategy with costs : SoftBank\\
Figure 17 : Log capital processes of Markov\\
\qquad \qquad \quad strategy with costs : IHI\\
Figure 18 : Log capital processes of Markov\\
\qquad \qquad \quad strategy with costs : Sony
\end{flushleft}
\end{minipage}
\begin{minipage}{.5\linewidth}
\begin{center}
\includegraphics[width=8cm,height=7cm]{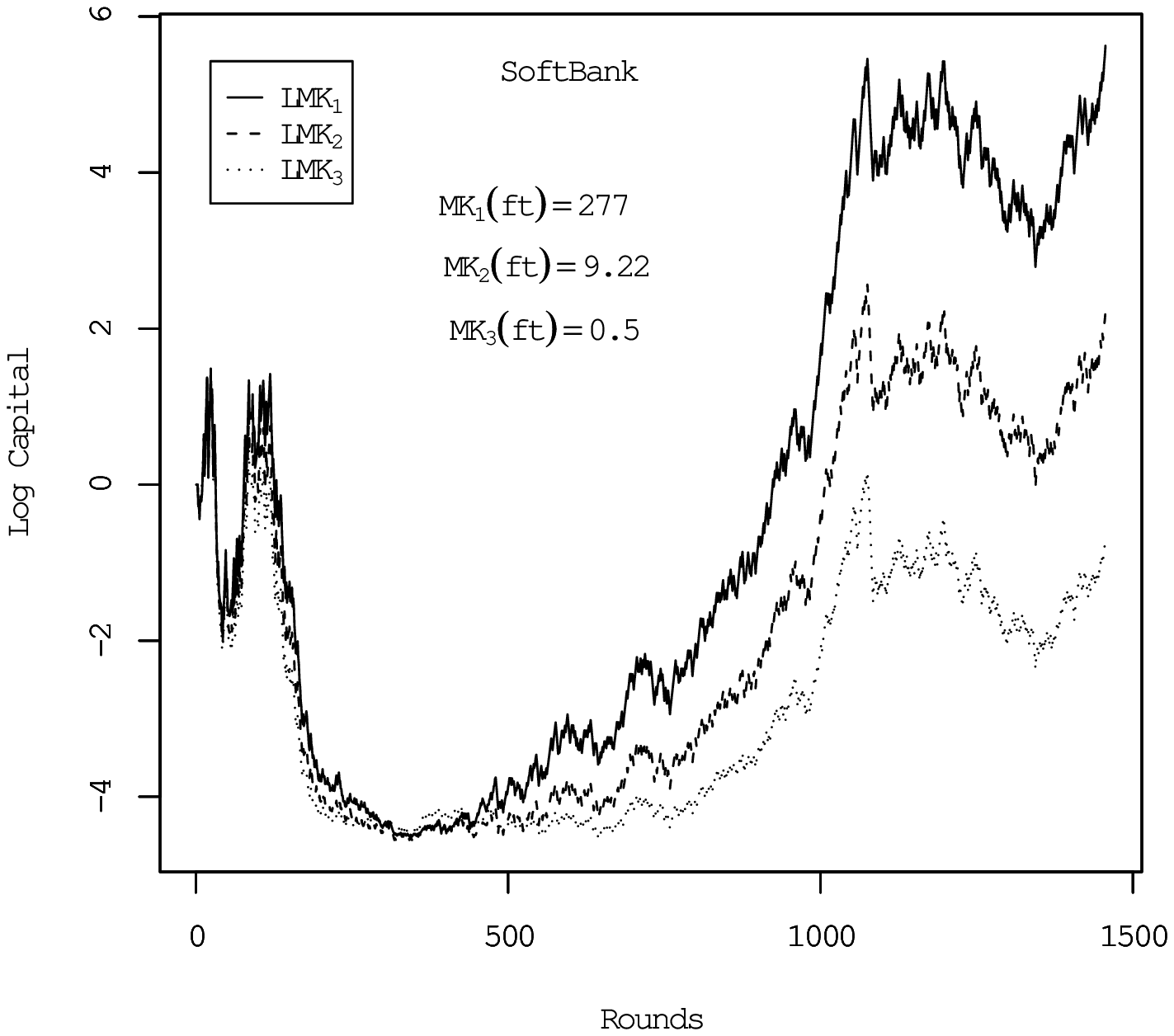}
\vspace*{-11mm}   
\caption{Log Markov capitals with costs}
\label{fig:16}
\end{center}
\end{minipage}
\begin{minipage}{.5\linewidth}
\begin{center}
\includegraphics[width=8cm,height=7cm]{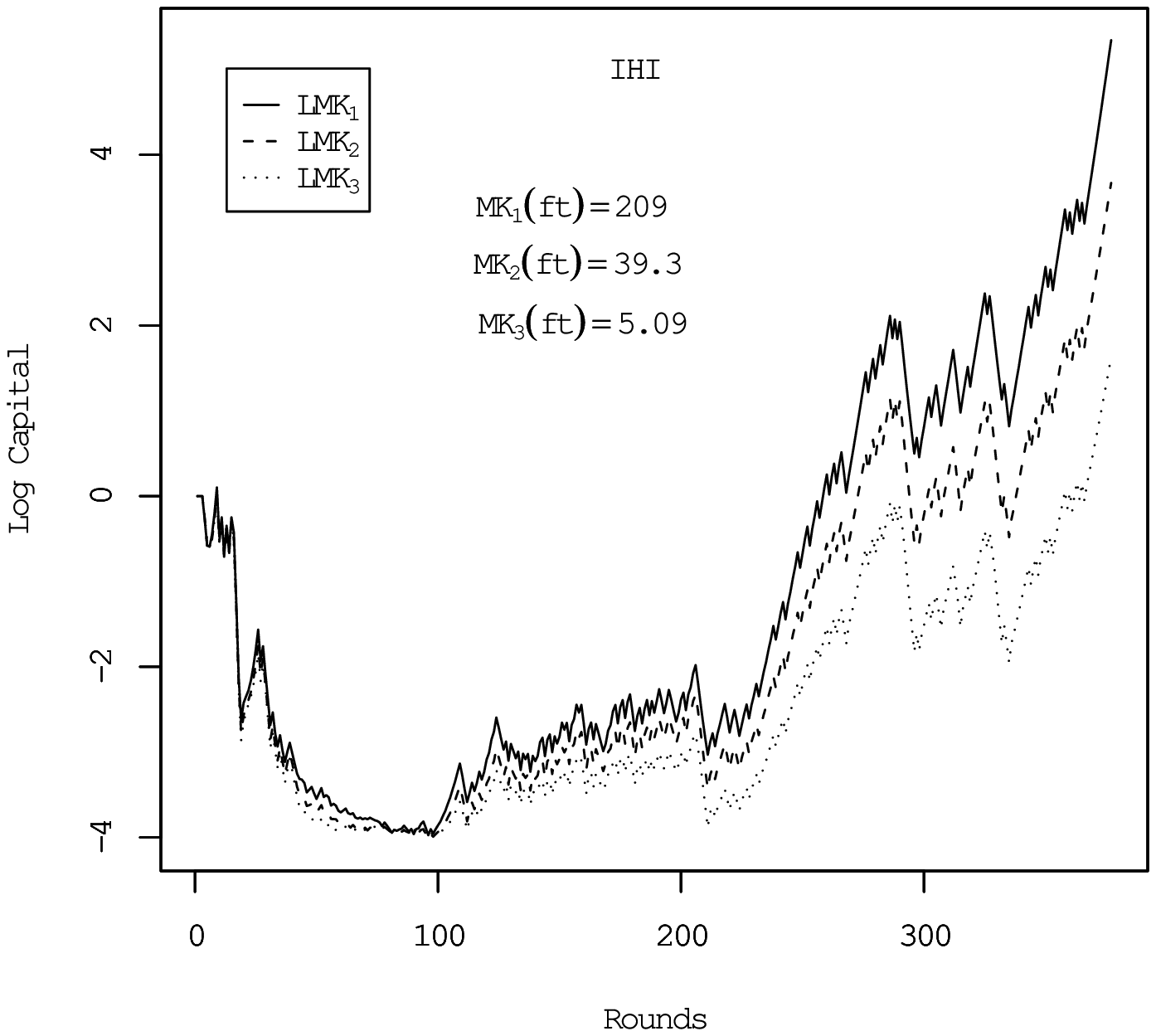}
\vspace*{-11mm}   
\caption{Log Markov capitals with costs}
\label{fig:17}
\end{center}
\end{minipage}
\begin{minipage}{.5\linewidth}
\begin{center}
\includegraphics[width=8cm,height=7cm]{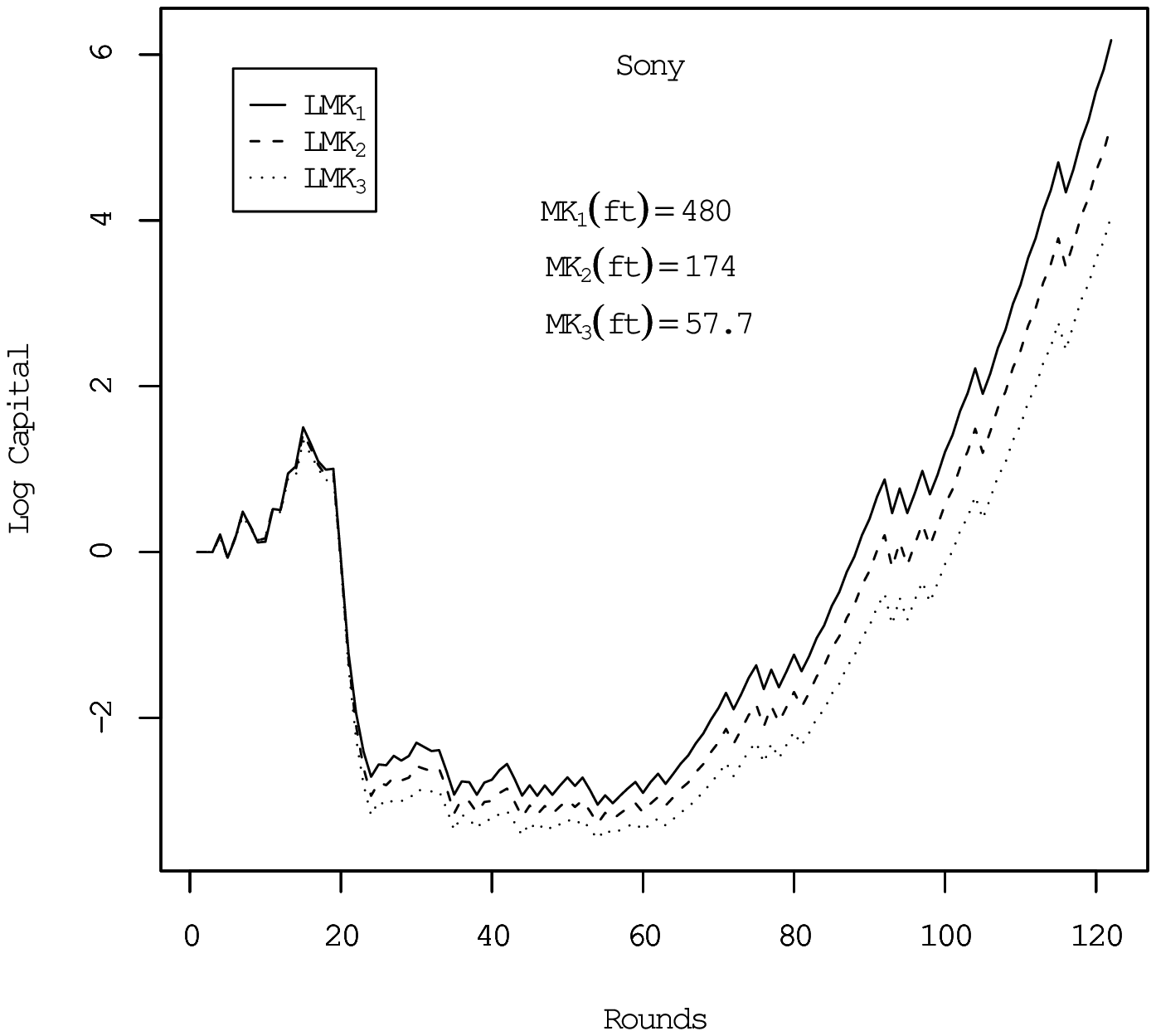}
\vspace*{-11mm}   
\caption{Log Markov capitals with costs}
\label{fig:18}
\end{center}
\end{minipage}
\end{figure}

\clearpage

\begin{table}[htbp]
\begin{center}
\label{tab:2}
\caption{Capital values for three stock minute prices with costs 
$(\eta = 2^{-8})$}
\ 

\begin{tabular}{c|r|rr|rr|rr}
\hline
& fn & $\textrm{MK}_1$(ft) & hn & $\textrm{MK}_2$(ft) & hn & $\textrm{MK}_3$(ft) & hn \\
\hline
SoftBank & 1456 & 277 & 326 & 9.22 & 551 & 0.50 & 636 \\
IHI & 377 & 209 & 25 & 39.3 & 103 & 5.09 & 129 \\
Sony & 122 & 480 & 10 & 174 & 20 & 57.7 & 27 \\
\hline
\end{tabular}
\end{center}
\end{table}

\begin{table}[htbp]
\begin{center}
\label{tab:3}
\caption{Capital values for three stock minute prices with costs 
$(\eta = 2^{-9})$}
\ 

\begin{tabular}{c|r|rr|rr|rr}
\hline
& fn & $\textrm{MK}_1$(ft) & hn & $\textrm{MK}_2$(ft) & hn & $\textrm{MK}_3$(ft) & hn \\
\hline
SoftBank & 1456 & 192051 & 431 & 2694 & 670 & 47.5 & 771 \\
IHI & 377 & 9907 & 8 & 403 & 27 & 22.3 & 65 \\
Sony & 122 & 2759414 & 6 & 462385 & 8 & 78155 & 12 \\
\hline
\end{tabular}
\end{center}
\end{table}

\begin{table}[htbp]
\begin{center}
\label{tab:4}
\caption{Critical values of unit costs for three stock minute prices}
\ 

\begin{tabular}{c|rrr|rrr}
\hline
$\eta = \log (1+\delta)$ & & $\eta = 2^{-8}$  & & &  $\eta = 2^{-9}$ \\
\hline
& $\textrm{MK}$(ft) & hn & $c^*$ & $\textrm{MK}$(ft) & hn & $c^*$ \\
\hline
SoftBank & 0.50 & 636 & $0.05\delta$ & 0.20 & 897 & $0.08\delta$ \\
IHI & 0.81 & 144 & $0.07\delta$ & 0.44 & 91 & $0.08\delta$ \\
Sony & 0.74 & 48 & $0.13\delta$ & 0.57 & 40 & $0.23\delta$ \\
\hline
\end{tabular}
\end{center}
\end{table}

From a practical point of view, the interesting question is not if the markets are efficient, but rather it is whether there are profitable trading opportunities after taking transaction costs into account. 
Thus we examine the effect of transaction costs on the Markov strategy. 
Let Investor pay the transaction cost at round $n$ specified by
\begin{align*}
C_n = c|\Delta M_n|S_{n-1} = c|M_n - M_{n-1}|S_{n-1},\quad c > 0,
\end{align*}
which is proportional to the traded monetary amount. 
Then his capital at the end of round $n$ is expressed as
\begin{align*}
{\cal K}_n &= {\cal K}_{n-1} + M_n(S_n - S_{n-1}) - c|\Delta M_n|S_{n-1}\\
&= {\cal K}_{n-1}(1 + \mu_{n-1}ds_n + \beta_n(ds_n - c\ \textrm{sgn} \beta_n)),
\end{align*}
where
\begin{align*}
ds_n = \frac{S_n - S_{n-1}}{S_{n-1}},\quad 
\mu_{n-1} = \frac{M_{n-1}S_{n-1}}{{\cal K}_{n-1}},\quad 
\beta_n = \frac{\Delta M_nS_{n-1}}{{\cal K}_{n-1}}. 
\end{align*}
At the start of round $n$ Investor again considers the maximization of
\begin{align*}
g_n(\beta_n) = E_{Q^\pm(ds_n|ds_{n-1})}[\log (1 + \mu_{n-1}ds_n + \beta_n(ds_n - c\ \textrm{sgn} \beta_n))|ds_{n-1}]
\end{align*}
with respect to $\beta_n \in {\mathbb R}$, where 
$Q^+(ds_n|ds_{n-1} > 0)$ and $Q^-(ds_n|ds_{n-1} < 0)$ denote the conditional beta-binomial distributions for the Market's limit order type moves. This results in his strategy summarized by the following scheme.
\begin{align*}
&g'_n(+0) > 0\ \Rightarrow\ \textrm{buy the asset of amount}\ 
\beta_n^+{\cal K}_{n-1}\ \textrm{determined by}\ g_n'(\beta_n^+) = 0 , \\
&g'_n(-0) < 0\ \Rightarrow\ \textrm{sell the asset of amount}\ 
-\beta_n^-{\cal K}_{n-1}\ \textrm{determined by}\ g_n'(\beta_n^-) = 0 , \\
&g'_n(+0) \le 0\ \&\ g'_n(-0) \ge 0\ \Rightarrow\ 
\textrm{hold the asset without trade} .
\end{align*}
Figures 16-18 compare the log capital processes of Markov strategy with transaction costs for SoftBank, IHI and Sony. In each figure, three kinds of log capital processes $LMK_1, LMK_2$ and $LMK_3$ are exhibited, which correspond to the unit costs $c = 0.01\delta, c = 0.03\delta$ and $c= 0.05\delta$, respectively. Table 2 $(\eta = 2^{-8})$ and Table 3 $(\eta = 2^{-9})$ show the capital values $\textrm{MK}_1$(ft), $\textrm{MK}_2$(ft), $\textrm{MK}_3$(ft), and the numbers of holding rounds  hn in the total fn rounds are also provided. Table 4 lists the critical values $c^*$ of unit costs which reduce the final amounts of Markov capitals MK(ft) to $\textrm{MK}$(ft) $< 1$ (less than the initial capital) for the first time.  We see that for SoftBank and IHI, the Markov capitals are severely reduced by $c^*$ which is even within ten percent of $\delta$ (limit order size), but these may depend on brands as seen in Sony.

\section{Some discussions}
\label{sec:discussion}

The efficient market hypothesis has been continuously discussed by many researchers.
Some well known researchers (e.g.\ \cite{rubinstein-2001},\cite{malkiel-2005})
argue that EMH is by and large true despite some observed irregularities.
However, from the literature we observe the following general tendencies.
1) 
Random walk hypotheses, without time-scale transformation, 
seem to be more often rejected than  accepted. 
2) There is only few literature dealing directly with
the martingale hypothesis.  This is probably due to the non-parametric
nature of the hypothesis and the difficulty in statistical modeling. 
3) As advocates of EMH argue, even professional investors
do not have  effective investing strategies outperforming the market.
4) Some classes of martingale models, especially those with
varying volatility, have been proposed and fitted to empirical data. But
their theoretical implications for effective investing strategies are not clear.

In this paper we presented a simple general method for directly testing the
hypothesis of martingale, by using limit order type investing strategies
in asset trading games.   The reciprocal of the capital process
of an investing strategy can be used as a $p$-value of test statistic
for testing the hypothesis of martingale property.  By our Markov type strategy
we have shown that the martingale property of some Japanese stocks 
are rejected with very small $p$-values.

It should be noted that our numerical experiments may not be realistic
for two reasons.  First, there is the problem of transaction cost. 
To test the martingale hypothesis, we have used a high frequency
limit-order type strategy. In actual markets high frequency trading
incurs a high trading cost  and we have examined this point in Section \ref{sec:numerical} by proportional-type transaction costs.
We found that they greatly reduce the capital, although this may depend on brands.
Second point is the reaction of the price to the amount of trading.  
In the usual measure-theoretic assumption, the amount of trading
does not affect the price process.  However in actual markets, 
large demand from the traders will immediately affect the price, %
thwarting the possibility of indefinitely large gain. These points may affect the
practical applicability of the proposed strategies, but they do not affect
the conclusion that the martingale hypothesis is rejected.

For the numerical experiments of Section \ref{sec:numerical} we have
tried several grid sizes (common to all processes) 
and showed a grid size $\eta = 2^{-8}$ 
which exhibits a significant result (nominal level of 0.1\%).  Therefore
there is a problem of multiple testing and to adjust for multiple
testing we can use the Bonferroni correction.  Since we have used 
Bayesian type simple strategy and its Markov type variant with only several
choices of grid sizes, the conclusion of Section \ref{sec:numerical} is
clearly valid with 1\% significance level.

\end{document}